\documentclass[10pt,journal,compsoc]{IEEEtran}

\usepackage[T1]{fontenc}
\ifCLASSOPTIONcompsoc

\usepackage{cite}

\usepackage{bm}
\usepackage{amsmath}
\usepackage{tikz}
\usepackage{amsthm}
\usepackage{float}
\usepackage{setspace}
\usepackage{amssymb}
\usepackage{stfloats}
\usepackage{cite}
\usepackage{ragged2e}
\usepackage[ruled,vlined,linesnumbered]{algorithm2e}
\usepackage{amsfonts}
\usepackage{mathrsfs}
\usepackage{amsmath,amsthm}
\usepackage{array,booktabs}
\usepackage{subfigure}
\usepackage{multirow}
\usepackage{cuted}
\usepackage{multicol}
\usepackage{graphicx}
\usepackage{subfigure}
\usepackage{graphicx,xcolor,bm}
\usepackage{hyperref}
\usepackage{threeparttable}
\usepackage{dcolumn}
\usepackage{setspace}
\usepackage{makecell}
\usepackage{lipsum}
\usepackage{enumerate}
\usepackage{mathrsfs}
\usepackage{bbm}
\usepackage[bottom]{footmisc}

\newtheorem{Defn}{Definition}
\newtheorem{lem}{Lemma}

\newtheorem{thm}{Theorem}

\usepackage{cite,bm}
\graphicspath{{figures/}}
\begin{document}
	\title{Accelerating Stable Matching between Workers and Spatial-Temporal Tasks for Dynamic MCS: \\A Stagewise Service Trading Approach}
	
\author{Houyi Qi, Minghui Liwang, \IEEEmembership{Senior Member}, \IEEEmembership{IEEE}, Xianbin Wang, \IEEEmembership{Fellow}, \IEEEmembership{IEEE}, Liqun Fu, \IEEEmembership{Senior Member}, \IEEEmembership{IEEE}, Yiguang Hong, \IEEEmembership{Fellow}, \IEEEmembership{IEEE}, Li Li, \IEEEmembership{Member}, \IEEEmembership{IEEE}, and Zhipeng Cheng, \IEEEmembership{Member}, \IEEEmembership{IEEE}
 
	\thanks{H. Qi (houyiqi@tongji.edu.cn), M. Liwang (minghuiliwang@tongji.edu.cn), Y. Hong (yghong@iss.ac.cn), and L. Li (lili@tongji.edu.cn) are with the Shanghai Research Institute for Intelligent Autonomous Systems, and also with the State Key Laboratory of Autonomous Intelligent Unmanned Systems, Frontiers Science Center for Intelligent Autonomous Systems, Ministry of Education, Shanghai Key Laboratory of Intelligent Autonomous Systems, Department of Control Science and Engineering, Tongji University, Shanghai, China. X. Wang (xianbin.wang@uwo.ca) is with the Department of Electrical and Computer Engineering, Western University, Ontario, Canada. L. Fu (liqun@xmu.edu.cn) is with the School of Informatics, Xiamen University, Fujian, China. Z. Cheng (chengzp\_x@163.com) is with the School of Future Science and Engineering, Soochow University, Jiangsu, China.
		
		Corresponding author: Minghui Liwang
	}}

	\IEEEtitleabstractindextext{\vspace{-3.5mm}
		\begin{abstract}
			\justifying
{Designing effective incentive mechanisms in mobile crowdsensing (MCS) networks is crucial for engaging distributed mobile users (workers) to contribute heterogeneous data for various applications (tasks). In this paper, we propose a novel stagewise trading framework to achieve efficient and stable task-worker matching, explicitly accounting for task diversity (e.g., spatio-temporal limitations) and network dynamics inherent in MCS environments. This framework integrates both futures and spot trading stages. In the former, we introduce the \textbf{f}utures \textbf{t}rading-driven \textbf{s}table \textbf{m}atching and \textbf{p}re-\textbf{p}ath-\textbf{p}lanning mechanism (FT-SMP$^3$), which enables long-term task-worker assignment and pre-planning of workers' trajectories based on historical statistics and risk-aware analysis. In the latter, we develop the \textbf{s}pot \textbf{t}rading-driven \textbf{D}QN-based \textbf{p}ath \textbf{p}lanning and onsite \textbf{w}orker \textbf{r}ecruitment mechanism (ST-DP$^2$WR), which dynamically improves the practical utilities of tasks and workers by supporting real-time recruitment and path adjustment. We rigorously prove that the proposed mechanisms satisfy key economic and algorithmic properties, including stability, individual rationality, competitive equilibrium, and weak Pareto optimality. Extensive experiements further validate the effectiveness of our framework in realistic network settings, demonstrating superior performance in terms of service quality, computational efficiency, and decision-making overhead.}

		\end{abstract}
		
		\vspace{-1mm}
}
	
\maketitle
\IEEEdisplaynontitleabstractindextext
%
\IEEEpeerreviewmaketitle

\section{Introduction}
\IEEEPARstart{M}{obile} Crowdsensing (MCS) represents an effective solution for distributed information gathering, which creatively leverages the power of ubiquitous Internet of Things (IoT) devices embedded with connectivity, computing power and heterogeneous data\cite{SURVEY 1,SURVEY 2}. By enabling a cost-effective and dependable sensing paradigm, MCS offers significant support across diverse applications (also known as tasks), e.g., intelligent transportation, environmental monitoring, mobile healthcare, and urban public management \cite{SURVEY 2, Incentive 1,Incentive 2}.
To better engage heterogeneous smart devices (also known as workers) in performing distributed MCS tasks, designing appropriate incentive mechanisms is crucial, especially when workers are selfishness due to limited resources\cite{Incentive 2,Incentive 3,Privacy}. For example, a worker with its own local workload is generally unwilling to offer free data services to remote sensing tasks. To this end, a service trading market can be established over MCS networks for facilitating data sharing among workers (who can contribute data while receiving payments) and task owners (who can pay for data services in supporting their data-driven applications)\cite{SURVEY 1}.

To engage workers for data sharing, various incentive mechanisms have been developed to determine the optimal assignment of tasks to workers. However, most of them have overlooked several critical issues:

\noindent $\bullet$ \textit{Dynamic and uncertain nature of MCS networks:} Most existing works consider a static MCS environment\cite{RWork_stable 4,RWork_stable 5}, by assuming that workers offer changeless service capability, and can definitively arrive at the assigned task without any risk. However, real-world MCS networks are always dynamic and uncertain. For instance, a worker may encounter delay events on its way to the assigned task (e.g., traffic jam and accident). Besides, since the delivery of data generally relies on wireless communications, the channel quality is varying over time. These uncertain factors collectively impose significant challenges in responsive and beneficial data provisioning mechanism design.

\noindent $\bullet$ \textit{Diversity of trading modes:} Existing research typically concentrates on either \textit{spot trading} or \textit{futures trading} strategies\cite{Trading mode 1,Trading mode 2}, where the former represents a widely adopted onsite data sharing mode, aiming to map task owners to proper workers by analyzing the current network and market conditions. However, spot trading can lead to excessive overheads, e.g., delay and energy cost on decision-making\cite{Future Trading 1,Future Trading 2,Future Trading 3,Future Trading 4} as well as risks, e.g., failures in data delivery\cite{Future Trading 1,Future Trading 2}, which inevitably leads to negative impact on quality of experience (QoE). Motivated by these drawbacks, futures-based trading allows task owners and workers to make pre-decisions (e.g., long-term contracts) for future trading in advance, thus facilitating more responsive data delivery. Nevertheless, implementing futures trading can incur risks such as unsuitable contract terms, when having inaccurate prediction on uncertain factors.

\noindent $\bullet$ \textit{Diverse task demands:} Demands of tasks are always diverse, including factors such as age of information (AoI)\cite{Incentive 2}, geographic locations of their point of interest (PoI)\footnote{Many studies use POIs to represent the sensing regions of tasks. Since this paper considers discrete tasks while each of them has its own location, we use "location of task" instead of "location of PoI of task", for better readability.} \cite{RWork_dynamic 3}, and time windows (e.g., start/closing time of tasks)\cite{RWork_dynamic 1}. Nevertheless, many studies have overlooked these specific features, leading to unsatisfying service qualities and even task failures under realistic conditions.

{Motivated by the aforementioned challenges, this work proposes a novel stagewise trading framework that seamlessly integrates futures and spot trading within dynamic and uncertain MCS networks. This framework is designed to enable risk-aware and stable task-worker matching\footnote{Stable matching refers to the stability of task-worker matching, ensuring that under allocation constraints and preferences, no blocking pairs exist, and no participant has an incentive to deviate \cite{RWork_stable 5,DP2}. In contrast, the static environment refers to an MCS network with deterministic parameters. The key distinction lies in that the dynamic MCS environment is affected by uncertainties such as variable delays and time-varying channel quality.}, while simultaneously supporting responsive and cost-efficient service transactions. In the futures trading stage, we employ a novel mechanism called \textbf{f}utures \textbf{t}rading-driven \textbf{s}table \textbf{m}atching and \textbf{p}re-\textbf{p}ath-\textbf{p}lanning (FT-SMP$^3$), which adapts to the evolving statistical characteristics of MCS networks to determine \textit{long-term workers} for each task (i.e., workers who can sign a contract with the task), while also incorporating pre-planned service paths for these workers.
In the subsequent spot trading stage, we introduce the \textbf{s}pot \textbf{t}rading-driven \textbf{D}QN-based \textbf{p}ath planning and onsite \textbf{w}orker \textbf{r}ecruitment (ST-DP$^2$WR) mechanism as an efficient backup. It allows workers to quickly adjust their strategies and improve task completion in dynamic environments, while also assisting tasks with remaining budgets to recruit supplementary worker, further enhancing service quality.}

Nevertheless, integrating diverse trading modes within a unified framework can also introduce new technical challenges. For example, \textit{futures trading aims to obtain long-term contracts, which may cause risks due to market fluctuations.} These risks, such as unsatisfying utility of workers and tasks, and uncompletion of tasks caused by unexpected delays, can lead to performance degradation of MCS. To cope with this, \textit{we incorporate risk constraints to keep possible risks within an acceptable range, thus improving the market performance.} Moreover, \textit{through the spot trading mode as a good backup}, workers can make recommend decisions when encountering delays, thereby improve their utility. Meanwhile, task owners can recruit temporary workers as needed to further enhance service quality\footnote{This two-stage design philosophy finds analogues in practical systems. For instance, during peak hours or holidays, ride-hailing platforms often schedule drivers in advance through timeslot incentives to ensure availability—analogous to the \textit{futures stage}. However, due to cancellations or traffic delays, they also conduct real-time dispatching to reassign orders—analogous to the \textit{spot stage}. This two-phase strategy balances long-term planning with real-time responsiveness.}.
{To the best of our knowledge, this paper makes a pioneering effort in designing a series of matching mechanisms that explicitly incorporate diverse spatio-temporal factors within a two-stage service trading market, as delineated along the timeline. Key contributions are summarized as follows:

\noindent $\bullet$ Considering the inherent dynamics and uncertainties of MCS environments, we propose a novel futures-spot integrated service trading market, conceptualized from an unique timeline-oriented perspective. This market is designed to facilitate accurate and risk-aware matching between mobile workers and spatio-temporal tasks, by systematically analyzing multiple sources of uncertainty, including stochastic delay events and variable task durations induced by such delays.}

\noindent $\bullet$ During the futures trading stage, we propose FT-SMP$^3$ for recruiting long-term workers who are more likely to catch the closing time of certain tasks, while also predetermining proper paths for workers, as the guidance about when to execute their assigned tasks. More importantly, we thoroughly analyze and manage the potential risks that tasks and workers may confront. We demonstrate that FT-SMP$^3$ supports key properties such as matching stability, individual rationality, fairness, and non-wastefulness. Also, we verify that the matching satisfies both competitive equilibrium and weak Pareto optimality.

\noindent $\bullet$ During spot trading stage, to raise the number of tasks that can successfully be completed in time, we develop ST-DP$^2$WR, offering two key functions: \textit{i)} helping workers to complete as many tasks as possible, thereby maximizing workers' utilities; and \textit{ii)} helping task owners with remaining budgets recruit more temporary workers to enhance their obtained service quality. Our ST-DP$^2$WR can improve the overall efficiency and stability of service trading market. Also, matchings considered in ST-DP$^2$WR satisfy similar properties to that involved in FT-SMP$^3$.

\noindent $\bullet$ We conduct comprehensive experiments based on both real-world dataset and synthetic rational parameters to verify the performance of our mechanisms in terms of service quality, social welfare, running time, and the overhead of interactions among participants, while also demonstrating our supportiveness on crucial properties.

\section{Related Work}
	\begin{table}[b!] 
	\vspace{-0.75cm}
	{\footnotesize
		\caption{\footnotesize{A summary of related studies\\(LoT: Location of task, TW: Time window of tasks)}} \vspace{-0.6cm} 
		\begin{center}
			\setlength{\tabcolsep}{0.5mm}{
				\begin{tabular}{|c|c|c|c|c|c|c|c|c|}
					\hline
					\multirow{2}{*}{\textbf{Reference}} & \multicolumn{2}{c|}{\makecell[c]{\textbf{Environmental}\\ \textbf{attributes}}} & \multicolumn{2}{c|}{\textbf{Trading mode}}&\multicolumn{4}{c|}{\textbf{Task property}}\\ \cline{2-9} 
					&\makecell[c]{Static}&\makecell[c]{Dynamic}&\makecell[c]{Spot}&\makecell[c]{Futures}&AoI&\makecell[c]{LoT}&TW&\makecell[c]{Budget}\\ \hline
					\makecell[c]{\cite{Incentive 1,RWork_stable 7,RWork_stable 5}} &$\surd$& &$\surd$& & & & &$\surd$\\ \hline
					\makecell[c]{\cite{Incentive 2}} &$\surd$& &$\surd$& &$\surd$& & &$\surd$\\ \hline	
					\makecell[c]{\cite{Incentive 3,RWork_stable 8}} &$\surd$& &$\surd$& & & $\surd$& &$\surd$\\ \hline
					\makecell[c]{\cite{RWork_stable 4}} &$\surd$& &$\surd$& &$\surd$& & &\\ \hline
					\makecell[c]{\cite{RWork_stable 6}} &$\surd$& &$\surd$& & &$\surd$&$\surd$&$\surd$\\ \hline
					\makecell[c]{\cite{RWork_dynamic 1}} & &$\surd$&$\surd$& &$\surd$&$\surd$&$\surd$&$\surd$\\ \hline
					\makecell[c]{\cite{RWork_dynamic 2}} &&$\surd$&$\surd$& &$\surd$& & &$\surd$\\ \hline
					\makecell[c]{\cite{RWork_dynamic 3,RWork_dynamic 4}} & &$\surd$ &$\surd$& & &$\surd$ &$\surd$&$\surd$ \\ \hline
					\makecell[c]{\cite{SURVEY 2,RWork_dynamic 5}} & &$\surd$ &$\surd$& & &$\surd$ & &$\surd$ \\ \hline
					\cite{DP2} & &$\surd$&$\surd$&$\surd$& & & &$\surd$\\ \hline
					our work & &$\surd$&$\surd$&$\surd$&$\surd$&$\surd$&$\surd$&$\surd$\\ \hline
			\end{tabular}}
	\end{center}}
	
\end{table}
Existing efforts have been put forward to resource trading in MCS networks from different viewpoints.

\noindent
$\bullet$ \textit{Investigations regarding static MCS networks.} Most studies on task scheduling and worker recruitment mainly consider rather static MCS networks \cite{Incentive 1,Incentive 2,Incentive 3,RWork_stable 4,RWork_stable 5,RWork_stable 6}.
In\cite{Incentive 1}, \textit{Zhou et al.} studied the bi-objective optimization for MCS incentive mechanism design, to simultaneously optimize total value function and coverage function with budget/cost constraint.
In\cite{Incentive 2}, \textit{Cheng et al.} considered AoI and captured the conflict interests/competitions among workers, proposing a freshness-aware incentive mechanism.
\textit{Hu et al.\cite{Incentive 3}} investigated a game-based incentive mechanism to recruit workers effectively while improving the reliability and data quality.
In\cite{RWork_stable 4}, \textit{Xiao et al.} considered the freshness of collected data and social benefits in MCS incentive designs
A many-to-many matching model was constructed by \textit{Dai et al.} \cite{RWork_stable 5} to capture the interaction between tasks and workers under budget constraints.
In\cite{RWork_stable 6}, \textit{Tao et al.} employed a double deep Q-network with prioritized experience replay to address the task allocation problem.
In\cite{RWork_stable 7}, \textit{Zhou et al.} studied a two-stage incentive scheme that combines blockchain technology and trusted execution environment.
\textit{Zhao et al.} \cite{RWork_stable 8} investigated a multi-agent deep reinforcement learning (DRL)-based incentive mechanism to tackle the joint data sensing and computing issues.

\noindent
$\bullet$ \textit{Investigations regarding dynamic MCS networks.} Although previous studies have made certain contributions, real-world MCS networks are inherently dynamic, and workers can often face various uncertain events during the data collection and delivery process. Consequently, researchers gradually shifted their focus towards dynamic and uncertain MCS networks\cite{SURVEY 2,RWork_dynamic 1,RWork_dynamic 2,RWork_dynamic 3,RWork_dynamic 4,RWork_dynamic 5}. 
In \cite{SURVEY 2}, \textit{Zhang et al.} considered diverse sensing tasks, while proposing a dynamic worker recruitment mechanism for edge computing-aided MCS. 
\textit{Gao et al.} \cite{RWork_dynamic 1} studied a dynamic task pricing problem with diverse factors such as multiple requester queuing competitions, dynamic task requirements, and distinct waiting time costs.
In \cite{RWork_dynamic 2}, \textit{Ji et al.} proposed a quality-driven online task-bundling-based incentive mechanism to maximize the social welfare while satisfying the task quality demands.
By adopting cognitive bias and the reference effect, \textit{Li et al.} in \cite{RWork_dynamic 3} explained the principle of path-dependence, and proposed a task coverage promotion according to path-dependence in improving the coverage and effectiveness.
In \cite{RWork_dynamic 4}, \textit{Ding et al.} investigated dynamic delayed-decision task assignment to enhance both the task completion ratios and budget utilization, while decreasing the user singleness.
In \cite{RWork_dynamic 5}, \textit{Guo et al.} proposed a dual reinforcement learning (RL)-based online worker recruitment strategy with adaptive budget segmentation, to cope with trajectories.
{While the aforementioned studies have raised some interesting ideas, they primarily focus on onsite decision-making (e.g., spot trading), which may be susceptible to prolonged delays, heavy energy cost, and potential trading failures.} To tackle the above challenges, we had made early efforts in establishing a hybrid market with futures and spot trading modes \cite{DP2}. However, \cite{DP2} puts emphasis on a rather simple environment, which overlooks several key aspects, such as the uncertain delay event incurred on the way of a worker to the assigned task, the diverse demands and attributes of MCS tasks, such like their spatial-temporal characteristics. These complicated uncertain factors are now considered in this paper.
Moreover, the optimization in \cite{DP2} simply aims to find a proper mapping between workers and tasks, while in this paper, we also should design the path for each worker to catch the deadline of assigned tasks. Driven by the complicated optimization problems, we design efficient solutions, such as ant colony optimization-based method for path pre-planning, and DRL-based method for path update (whether to follow the pre-planned path or abandon a contractual task). As a summary, we focus on closer-to-real-world settings for MCS. A comparative summary of related studies is presented in Table 1, highlighting the key distinctions of our approach. In the following, we elaborate on the unique features and methodological advantages of our framework over existing methods, as outlined below.
 
\noindent $\bullet$ \textit{Consideration on diverse task demands and uncertainties:} Unlike traditional approaches with single task characteristic -- such as geographic location, time constraints, or budget limitations -- our method fully accounts for multi-dimensional task requirements, including AoI, location of PoIs, time windows, and budget constraint. Also, uncertain factors such as uncertain delay event, uncertain duration of the delay event, and time-varying channel quality are carefully modeled to capture the random nature of MCS networks. That is to say, this paper pays attention to a more realistic MCS environment with spatio-temporal dynamics.

\noindent $\bullet$ \textit{Way to handle uncertainties and risks:} In dynamic MCS networks, uncertainties in task execution -- such as task cancellations, delays on the way, and worker mobility -- can leave heavy impacts on task completion and resource utilization, which have always been overlooked in conventional methods. To address this issue, we incorporate risk-constrained optimization during the futures trading stage, ensuring that task assignments remain feasible under diverse uncertainties while maintaining market stability. More importantly, during the spot trading stage, we enable workers to update/revise their paths in response to dynamic conditions, facilitating flexible on-demand worker recruitment strategies.

\section{Overview and System Model}
\subsection{Overview}
\begin{figure*}[]
	\vspace{-0.2cm}
	\centering
	\includegraphics[width=1.8\columnwidth]{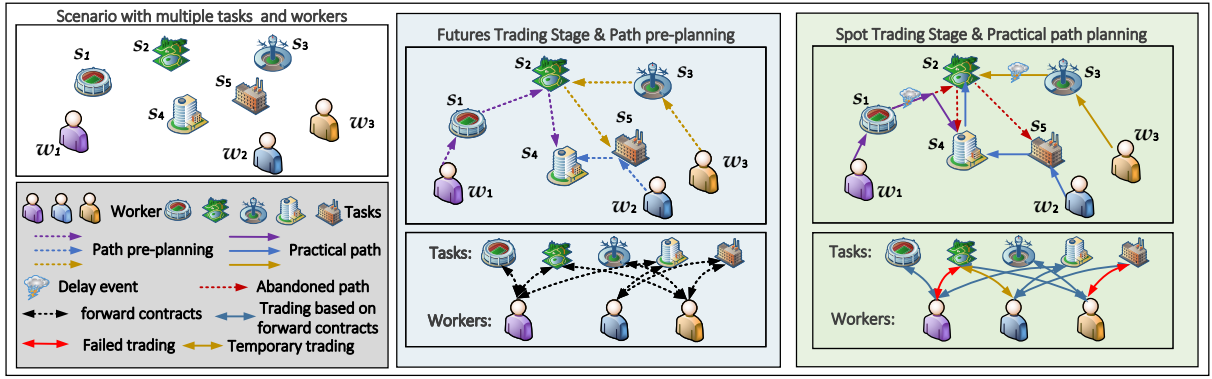}
	\caption{Framework and procedure in terms of a timeline associated with our proposed the stagewise trading framework in dynamic MCS networks.}
	\vspace{-0.4cm}
\end{figure*}
We are interested in a dynamic MCS network involving two key parties: \textit{i)} multiple sensing tasks collected in set $ \bm{S}=\left\{s_1,...,s_i,...,s_{|\bm{S}|}\right\} $, and \textit{ii)} multiple workers gathered via set $ \bm{W}=\left\{w_1,...,w_j,...,w_{|\bm{W}|}\right\} $. To capture the time-varying features, a practical transaction is discretized into $T$ timeslots, with index $t$, i.e., $t\in\left\{0,1,2,...,T\right\} $.
Also, to describe the dynamic and random nature of real-world MCS networks, we also involve the following uncertain factors\footnote{For the uncertain delay event and their duration, they actually leave real-time constraints on worker availability. Delays encountered during task execution may prevent workers from reaching the designated PoI within the specified MCS time window, increasing the uncertainty in worker availability. For time-varying channel quality, we model its fluctuations to account for variability in wireless communication performance.}:

\noindent
\textit{i) Uncertain delay event:} Worker $w_j$ may encounter delay events (e.g., traffic jam, traffic accident) on the way to task $s_i$. We denote the delay event as a random variable $\alpha_{i,j}$, following a Bernoulli distribution $\alpha_{i,j}\sim {\bf B}\left\{(1,0),(a_{i,j},1-a_{i,j})\right\}$. Specifically, $\alpha_{i,j}=1$ indicates that the $w_j$ encounters a delay event on its way to $s_i$, while $\alpha_{i,j}=0$, otherwise. 

\noindent
    \textit{ii) Uncertain duration of the delay event:} When encountering a delay event, worker $w_j$ will definitely spend a certain amount of time on it, which, to simplify our analysis, can be quantized by the number of timeslots denoted by $\tau^{\text{delay}}$. This duration follows a discrete uniform distribution, denoted by \(\tau^{\text{delay}} \sim \mathbf{U}(t^{\text{min}}, t^{\text{max}})\), e.g., $\tau^{\text{delay}}\in \{t^{\text{min}}, t^{\text{min}}+1,...,t^{\text{max}}-1, t^{\text{max}}\}$.

\noindent
    \textit{iii) Time-varying channel quality:} The network condition between a worker and a task owner is denoted by $\gamma_{i,j}$, reflected by the channel quality of wireless link between $w_j$ and the nearby access point (AP) for accessing $s_i$. The $\gamma_{i,j}$ follows a uniform distribution $\gamma_{i,j}\sim {\bf U}(\mu_1,\mu_2)$, where fluctuations can due to factors such as mobility of workers, and obstacles.

In this paper, we explore a novel matching-based stagewise service trading paradigm over dynamic MCS networks with the above uncertain factors. Interestingly, ``stagewise'' allows us to divide the whole trading process across multiple timeslots as two complementary stages. The former stage relies on a \textit{futures} trading mode, which encourages each task to recruit a set of workers with paid services, e.g., the corresponding payment (denoted by $ p_{i,j} $). While each worker can serve a vector of tasks, where the matched tasks are sorted into a path (sequence) by following factors such as the tasks' locations, time windows, payments, and costs. In addition, compensation from $w_j$ to $s_i$ (denoted by $ q_{i,j} $) may also be incurred when $w_j$ fails to complete $s_i$, e.g., $w_j$ encounters a set of delay events on the way to $s_i$, resulting in insufficient time in performing tasks. The above process can be implemented by facilitating mutually beneficial long-term contracts between workers and tasks, which are pre-signed (e.g., contract terms including $ p_{i,j} $ and $ q_{i,j} $) in advance to future practical transactions according to historical statistics. With these contracts in place, contractual workers and tasks can engage in practical transactions\footnote{A practical transaction refers to an actual service trading event between tasks and workers, where key issues are determined based on long-term contracts and the actual network conditions, including contract fulfillment, workers' payments, compensations for tasks, and the actual service paths of workers.} directly. As a complement, the latter stage relies on a spot trading mode, in which workers and tasks with long-term contracts are expected to fulfill their obligations accordingly. Since workers may fail to complete contractual tasks due to dynamic nature of MCS networks, tasks can further employ temporary workers when experiencing unsatisfying service quality.
Our major goal in this paper\textit{ is to establish fast, stable, and efficient worker recruitment as well as task execution paths, accounting for diverse task demands in dynamic and uncertain MCS networks.}

Fig. 1 depicts a schematic of the stagewise service trading market.
 For instance, workers and tasks first sign proper long-term contracts with pre-planned paths (the intermediate two boxes of Fig. 1). During a practical transaction (the right two boxes of Fig. 1), worker $w_1$ gives up $s_2$ to proceed to the next task. Similarly, worker $w_3$ decides to continue $s_2$, but fails to complete $s_5$. Also, due to uncertain factors, long-term workers may not be able to catch the deadline of assigned tasks. Facing this case, task owners with remaining budgets can recruit temporary workers (e.g., in Fig. 1, $s_2$ recruits $w_2$).

\subsection{Basic Modeling}
\textbf{Modeling of Tasks:} The attribute of each MCS task $s_i\in \bm{S}$ is represented by a 6-tuple $\langle t_i^{\text{b}}, t_i^{\text{e}}, B_i, Q_i^{\text{D}}, l_i^{\text{s}}, d_i \rangle$, where $t_i^{\text{b}}$ and $t_i^{\text{e}}$ define the start and closing time of $s_i$ (in terms of timeslot, e.g., $s_i$ starts at the $t_i^{\text{b}}$-th timeslot), collectively form a time window. Additionally, $B_i$ denotes the budget of task $s_i$, constraining the number of employable workers; $Q_i^{\text{D}}$ represents the desired service quality of $s_i$ (e.g., its required AoI); the location of $s_i$ is given by $l_i^{\text{s}} = (lo_i^{\text{s}}, la_i^{\text{s}}) $, where $lo_i^{\text{s}}$ and $la_i^{\text{s}}$ represent the longitude and latitude of $s_i$, respectively. The data size that each worker provides to $s_i$ is denoted by $d_i$ (e.g., bits).

\noindent\textbf{Modeling of Workers:} The attribute of a worker $w_j\in \bm{W}$ is described by a 6-tuple $ \langle e_j^{\text{c}}, e_j^{\text{D}}, e_j^{\text{t}}, e_j^{\text{m}}, f_j, v_j, l_j^{\text{w}}\langle t\rangle \rangle$, where $e_j^{\text{c}}$ (e.g., dollar/timeslot), $e_j^{\text{D}}$ (e.g., dollar/timeslot), $e_j^{\text{t}}$ (e.g., J/timeslot), and $e_j^{\text{m}}$ (e.g., dollar/timeslot) indicate the cost consumed in each timeslot for data collection, delay event, data transmission, and traveling to the target task, respectively. Note that for $e_j^{\text{c}}$, $e_j^{\text{D}}$, and $e_j^{\text{m}}$, we consider a monetary representation to better describe a resource trading market. Besides, $f_j$ (e.g., bits) represents the size of data collected by worker $w_j$ in each timeslot, and the moving speed of worker $w_j$ is denoted by $v_j$ (e.g., meter/timeslot). The position of $w_j$ at timeslot $t$ is represented by $l_j^{\text{w}}\langle t\rangle = ( lo_j^{\text{w}}\langle t\rangle, la_j^{\text{w}}\langle t\rangle) $, where $lo_j^{\text{w}}\langle t\rangle$ and $la_j^{\text{w}}\langle t\rangle$ denote the longitude and latitude of worker $w_j$, respectively.

\vspace{-3mm}\section{Proposed FT-SMP$^3$}
We next introduce FT-SMP$^3$, a novel approach specifically designed to accommodate the characteristics of the futures trading stage. At this time, our objective is to establish an efficient and stable matching between workers and diverse demands of tasks, and thereby:
\textit{i)} forming a set of long-term contracts to guide future transactions;
\textit{ii)} determining the potential paths that workers may follow.
{Crucially, FT-SMP$^3$ is executed only once during the offline futures trading stage, prior to actual task execution. This design significantly alleviates the computational burden on online decision-making (for example, the associated path planning process can also be offloaded to resource-rich nodes, e.g., cloud servers or capable edge devices during off-peak periods, thereby eliminating runtime overhead on the online system. Regarding this issue, we will not engage in any detailed discussion here in this paper).}
\vspace{-2mm}\subsection{Key Modeling}

We first define the matching between tasks and workers:

\noindent
$\bullet$ $\varphi(s_i)$: the set of workers recruited for processing task $ s_i $, i.e., $ \varphi(s_i) \in \bm{W} $;

\noindent
$\bullet$ $ \varphi(w_j) $: the vector of sensing tasks assigned to worker $ w_j $, i.e., $ \varphi(w_j) \in \bm{S} $;

Notably, we replace notations $p_{i,j}$, $q_{i,j}$ during futures trading stage by $p^F_{i,j}$, $q^F_{i,j}$, to avoid possible confusion with spot trading stage.

\subsubsection{Utility, expected utility, and risk of workers}
As performing tasks can incur costs on workers, we consider four types of costs for worker $w_j$ to complete task $s_i$:

\noindent
\textit{i) Moving cost.} The Euclidean distance between task $s_i$ and worker $w_j$ at timeslot $t$ is calculated as $ Ed_{i,j}\langle t\rangle = \left((lo_j^{\text{w}}\langle t\rangle - lo_i^{\text{s}})^2 + (la_j^{\text{w}}\langle t\rangle - la_i^{\text{s}})^2\right)^{\frac{1}{2}} $. Besides, we can have the time for $w_j$ moving to $s_i$ as $ \tau_{i,j}^{\text{move}}\langle t\rangle = \frac{Ed_{i,j}\langle t\rangle}{v_j} $. Accordingly, the moving cost at timeslot $t$ can be expressed by
\begin{equation}\label{key}\tag{1}
	c_{i,j}^{\text{move}}\langle t\rangle= e_j^{\text{m}} \tau_{i,j}^{\text{move}}\langle t\rangle=\frac{e_j^{\text{m}} Ed_{i,j}\langle t\rangle}{v_j} .
\end{equation}

\noindent\textit{ii) Sensing cost.} The time that worker $w_j$ spends on collecting data for task $s_i$ is defined as $ \tau_{i,j}^{\text{sense}} = \frac{d_i}{f_j} $. Thus, we calculate the data collection cost as
\begin{equation}\label{key}\tag{2}
	c_{i,j}^{\text{sense}}=e_j^{\text{c}}\tau_{i,j}^{\text{sense}}= \frac{e_j^{\text{c}} d_i}{ f_j}.
\end{equation}

\noindent\textit{iii) Transmission cost.} The time consumed by worker $ w_j $ transmitting sensing data to $s_i$ as
\begin{equation}\label{key}\tag{3}
	\tau_{i,j}^{\text{tran}}= \frac{d_i}{W\text{log}_2 \left(1+e_j^{\text{t}} \gamma_{i,j}\right)},
\end{equation}
where $ W $ is the bandwidth of wireless communication links, and $ e_j^{\text{t}} \gamma_{i,j} $ indicates the received signal noise ratio (SNR) of $ s_i $. The transmission cost is expressed by
\begin{equation}\label{key}\tag{4}
	c_{i,j}^{\text{tran}}=e_j^{\text{t}} \tau_{i,j}^{\text{tran}}= \frac{e_j^{\text{t}} d_i}{W\text{log}_2 \left(1+e_j^{\text{t}} \gamma_{i,j}\right)},
\end{equation}

\noindent\textit{iv) Cost incurred by delay event.} Since a worker $w_j$ may encounter delay events to task $s_i$, we define the delay time for worker $w_j$ traveling to $s_i$ as $\tau_{i,j}^{\text{D}}\langle t^{\text{ini}} \rangle=\sum_{n=1}^{\tau_{i,j}^{\text{move}}\langle t^{\text{ini}} \rangle}\alpha_{i,j}\tau^{\text{delay}}_n$, where $ t_{i,j}^{\text{ini}} $ represents the initial timeslot for $ w_j $ to set off for task $ s_i\in\varphi(w_j) $. The cost incurred by a delay event can thus be calculated as
\begin{equation}\label{key}\tag{5}
	c_{i,j}^{\text{D}}\langle t^{\text{ini}} \rangle=e_j^{\text{D}}\tau_{i,j}^{\text{D}}\langle t^{\text{ini}} \rangle.
\end{equation}

Accordingly, the overall cost on $ w_j $ performing $ s_i $ is
\begin{equation}\label{key}\tag{6}
	c_{i,j}\langle t^{\text{ini}} \rangle=\mathbb{V}_1\left(c_{i,j}^{\text{move}}\langle t^{\text{ini}} \rangle+c_{i,j}^{\text{D}}+c_{i,j}^{\text{sense}}+c_{i,j}^{\text{tran}}\right),
\end{equation}
where $\mathbb{V}_1$\footnote{The unit monetary cost $\mathbb{V}_1$ is a tunable system parameter that can reflect different economic contexts, such as electricity pricing or regional incentive policies. Other coefficients (e.g., $e_j^{\text{m}}$, $e_j^{\text{t}}$) are device-dependent and can be configured flexibly.} is the unit monetary cost/Joule (e.g., dollar/J), and the corresponding task completion time is given by
\begin{equation}\label{key}\tag{7}
	\tau_{i,j}\langle t^{\text{ini}} \rangle=\tau_{i,j}^{\text{move}}\langle t^{\text{ini}} \rangle+\tau_{i,j}^{\text{D}}\langle t^{\text{ini}} \rangle+\tau_{i,j}^{\text{sense}}+\tau_{i,j}^{\text{tran}}.
\end{equation}

Since worker $ w_j $ may encounter delay events and accordingly fails to complete a task in time, we use $ \beta_{i,j} $ to describe whether $ w_j $ completes $ s_i $ during practical transactions, as given by
\begin{equation}\label{key}\tag{8}
	\beta_{i,j}=\left\{
	\begin{aligned}
		&0,~~w_j \text{ breaks a contract with } s_i\\
		&1,~~w_j \text{ fulfills a contract with } s_i.
	\end{aligned} \right.
\end{equation}
The utility of worker $w_j$ involves three key components: \textit{i)} the overall payment minus the cost on performing tasks, \textit{ii)} the service cost which has been consumed on $w_j$, while confronting a failure in task completion, and \textit{iii)} the penalty for failing to complete tasks. Accordingly, the utility of worker $w_j$ is given by
\begin{equation}\label{key}\tag{9}
	\begin{aligned}
		&U^W (w_j,\varphi(w_j))=\sum_{s_i\in\varphi(w_j)}\beta_{i,j}\left(p^F_{i,j}-c_{i,j}\langle t_{i,j}^{\text{ini}} \rangle\right)\\&-\sum_{s_i\in\varphi(w_j)}(1-\beta_{i,j})\left(c_{i,j}^{\text{part}}\langle t_{i,j}^{\text{ini}} \rangle+q^F_{i,j}\right),
	\end{aligned}
\end{equation}
where $c_{i,j}^{\text{part}}\langle t_{i,j}^{\text{ini}} \rangle$ indicates the costs incurred when worker $w_j$ has made efforts but task $s_i$ still fails. This cost structure is similar to that of $c_{i,j}\langle t_{i,j}^{\text{ini}} \rangle$, including movement costs, delay event costs, sensing costs, and transmission costs. The specific value of \( c_{i,j}^{\text{part}}\langle t_{i,j}^{\text{ini}} \rangle \) is determined by the timeslot in which worker \( w_j \) decides to abandon task \( s_i \).
Apparently, uncertain factors can impose challenges to obtain the practical value of (9) in our designed futures trading stage. Instead, we are interested in its expection, as shown below
\begin{equation}\label{key}\tag{10}
	\begin{aligned}
		&\overline{U^W}(w_j,\varphi(w_j))=\text{E}[{U^W}(w_j,\varphi(w_j))]\\=&\sum_{s_i\in\varphi(w_j)}\text{E}[\beta_{i,j}]\left(p^F_{i,j}-\text{E}[c_{i,j}\langle t_{i,j}^{\text{ini}} \rangle]\right)-\\&\sum_{s_i\in\varphi(w_j)}\left(1-\text{E}[\beta_{i,j}]\right)\left( \text{E}[c_{i,j}^{\text{part}}\langle t_{i,j}^{\text{ini}} \rangle]+q^F_{i,j}\right),
	\end{aligned}
\end{equation}
where derivations of $ \text{E}\left[\beta_{i,j}\right] $, $ \text{E}\left[c_{i,j}\langle t_{i,j}^{\text{ini}} \rangle\right] $, and $ \text{E}\left[c_{i,j}^{\text{part}}\langle t_{i,j}^{\text{ini}} \rangle\right] $ are detailed in Appx. B.1.

A futures trading can generally bring both benefits and risks, as uncertainties may lead to losses to participants. Thus, we evaluate two specific risks for each worker:

\noindent\textit{i) The risk of receiving an unsatisfying utility:} Each worker $w_j \in \bm{W}$ serving task $s_i \in \varphi(w_j)$ faces the risk of obtaining an unsatisfying utility (e.g., when the value of $U^W(w_j,s_i)$ turns negative) during a practical transaction. This risk is defined by the probability that the utility $U^W(w_j,s_i)$ falls below a tolerable value $u_{\text{min}}$, given by:
\begin{equation}\label{key}\tag{11}
	R^W_1 (w_j ,s_i)=\text{Pr}\left(U^W (w_j,s_i)<u_{\text{min}}\right),~\forall s_i\in \varphi(w_j),
\end{equation}
where $ u_{\text{min}} $ is a positive value approaching to 0.

\noindent\textit{ii) The risk on failing to complete the task:} The time for a worker $w_j$ to complete task $s_i \in \varphi(w_j)$ may be insufficient due to possible delay events, as defined by
\begin{equation}\label{key}\tag{12}
	R_2^W (w_j ,s_i)=\text{Pr}\left(\beta_{i,j}=0\right),~~~\forall s_i\in \varphi(w_j).
\end{equation}

The aforementioned risks should be managed within an acceptable range, otherwise, worker $w_j$ will not sign a long-term contract with task $s_i$ during the futures trading stage.

\subsubsection{Utility, expected utility, and risk of tasks}
Considering the importance of freshness in MCS data, we use AoI as a key metric to assess service quality. AoI represents the elapsed time since the most recently received data were generated, and it continuously increases until a new update can be received. The lifecycle of data generally starts from collection (generation), undergoes a transmission period, and finally reaches the task owner for utilization.
	
Inspired by existing literature \cite{RWork_stable 4, Aoi 1, Aoi 2}, we first define the AoI of sensing data provided by worker \( w_j \) for task \( s_i \). Let \( t_{i,j}^{\text{gen}} \) denote the timeslot when worker \( w_j \) starts collecting data for task \( s_i \). The total AoI is obtained by summing over all timeslots from data generation, until it is fully transmitted and received by the MCS task, as given by 
	\begin{equation}\label{eq:aoi}\tag{13}{
		\begin{aligned}
age_{i,j}=&\sum_{t^\prime=t_{i,j}^{\text{gen}}}^{t_{i,j}^{\text{gen}}+\tau_{i,j}^{\text{sense}}+\tau_{i,j}^{\text{tran}}} \left(t^\prime-t_{i,j}^{\text{gen}}\right)=\sum_{t^\prime=0}^{\tau_{i,j}^{\text{sense}}+\tau_{i,j}^{\text{tran}}} t^\prime\\&=\frac{(\tau_{i,j}^{\text{sense}}+\tau_{i,j}^{\text{tran}})^2+\tau_{i,j}^{\text{sense}}+\tau_{i,j}^{\text{tran}}}{2}.
		\end{aligned}}
	\end{equation}
In particular, this summation accounts for all intermediate timeslots between \( t_{i,j}^{\text{gen}} \) (when the data is first generated) and \( t_{i,j}^{\text{gen}} + \tau_{i,j}^{\text{sense}} + \tau_{i,j}^{\text{tran}} \) (when the data is finally received). The term \( (t^\prime - t_{i,j}^{\text{gen}}) \) represents the instantaneous AoI at each timeslot.
	
{Accordingly, the average AoI of the sensing data that worker \( w_j \) contributes to task \( s_i \) is computed as $AGE_{i,j}=\frac{age_{i,j}}{\tau_{i,j}^{\text{sense}}+\tau_{i,j}^{\text{tran}}}=\frac{\tau_{i,j}^{\text{sense}}+\tau_{i,j}^{\text{tran}}+1}{2}$, where \( \tau_{i,j}^{\text{sense}}+\tau_{i,j}^{\text{tran}} \) represents the total time taken from data generation to its reception by MCS tasks.} To assess the overall service quality of task \( s_i \), we define \( Q(s_i, \varphi(s_i)) \) to denote the freshness of data contributions from all assigned workers\footnote{Although (14) adopts the reciprocal form $1/AGE_{i,j}$, which theoretically preserves a non-zero contribution from high-AoI workers. Thus, this design does not compromise the system's control over data freshness due to the joint effect of the following mechanisms:  
\textit{i)} As $AGE_{i,j}$ increases, its contribution rapidly approaches zero, making the impact of stale data practically negligible;  
\textit{ii)} The risk-aware constraint in (17) ensures that tasks reject worker combinations whose service quality falls below the acceptable threshold $Q_i^{\text{D}}$, thereby filtering out low-quality data sources at the system level;  
\textit{iii)} In the knapsack-based worker selection process (Step~3 of Alg.~2), high-AoI workers tend to exhibit low utility and are thus excluded from the optimal solution.
Therefore, even if a small theoretical contribution remains, the proposed design—when combined with the overall mechanism—can effectively suppress stale data sources and ensure high-quality long-term services.
}:
	\begin{equation}\label{eq:quality}\tag{14}
		Q\left(s_i,\varphi(s_i)\right)=\sum_{w_j\in \varphi(s_i)} \frac{1}{AGE_{i,j}},
	\end{equation}
which ensures the worker who provides fresher data (lower AoI) can contribute more to the overall service quality.

Accordingly, the utility of task $ s_i $ consists of two key parts: \textit{i)} the obtained service quality; \textit{ii)} the remaining budget (a larger remaining budget further reflects a lower cost), as given by
\begin{equation}\label{key}\tag{15}
	\begin{aligned}
	&U^S (s_i,\varphi(s_i))=\mathbb{V}_2\sum_{w_j\in \varphi(s_i)} \frac{1}{AGE_{i,j}}+\\& \mathbb{V}_3 \left(\sum_{w_j\in \varphi(s_i )}\left((1-\beta_{i,j})q^F_{i,j}-\beta_{i,j}p^F_{i,j}\right)\right),
	\end{aligned}
\end{equation}
where $\mathbb{V}_2 $ and $\mathbb{V}_3$\footnote{
To prevent unintended utility inflation caused by breach-related penalties, parameter $\mathbb{V}_3$ is constrained within interval $[0,1]$, ensuring that the task-side utility from contract breach (i.e., $\mathbb{V}_3 q^F_{i,j}$) can stay below the corresponding loss on the worker side ($q^F_{i,j}$). This design guarantees that the total social welfare (the summation of utilities of both tasks and workers) strictly decreases after a breach, thereby eliminating incentives for system-level utility manipulation.
} are positive weighting coefficients. As uncertainties prevent us from getting the practical value of the task's utility directly, we consider its corresponding expectation as
\begin{equation}\label{key}\tag{16}
	\begin{aligned}
		&\overline{U^S}\left(s_i,\varphi(s_i)\right)=\text{E}[U^S \left(s_i,\varphi(s_i)\right)]\\=&\mathbb{V}_2 \frac{1}{\text{E}[AGE_{i,j}]}+\\& \mathbb{V}_3 \left(\sum_{w_j\in \varphi(s_i )}\left((1-\text{E}[\beta_{i,j}])q^F_{i,j}-\text{E}[\beta_{i,j}]p^F_{i,j}\right)\right),
	\end{aligned}
\end{equation}
where derivations of $ \text{E}[AGE_{i,j}] $ is detailed in Appx. B.2.

Similar to workers, task owners also face risks owing to market dynamics. This risk of $s_i$ is primarily associated with the inability to reach desired service quality, given by
\begin{equation}\label{key}\tag{17}
	\begin{aligned}
	&R^S (s_i,\varphi(s_i))= \text{Pr}\left(Q(s_i,\varphi(s_i))<Q_i^{\text{D}}\right).
	\end{aligned}	
\end{equation}
Apparently, a larger value of $ R^S (s_i,\varphi(s_i)) $ implies a higher risk on an unsatisfying quality. Thus, the task owner will not sign long-term contracts with workers in the futures trading stage when confronting an unacceptable risk.

\subsection{Key Definitions of Matching}
We next introduce our designed matching in futures trading stage, representing a unique many-to-many (M2M) matching tailored to the characteristics of uncertainties, upon considering multiple timeslots. More importantly, this matching is also crafted to cope with potential risks in dynamic MCS networks, distinguishing it from conventional matching mechanisms.
\begin{Defn}(M2M matching of FT-SMP$^3$)
	A M2M matching $ \varphi $ of our proposed FT-SMP$^3$ constitutes a mapping between task set $ \bm{S} $ and worker set $ \bm{W} $, which satisfies the following properties:
	
	\noindent
	$\bullet$ for each task $ s_{i} \in \bm{S},\varphi\left( s_i \right) \subseteq \bm{W} $,
	
	\noindent
	$\bullet$ for each worker $ w_{j} \in \bm{W}, \varphi\left( w_{j} \right) \subseteq \bm{S} $,
	
	\noindent
	$\bullet$ for each task $ s_i $ and worker $ w_j $, $ s_i\in\varphi(w_j)$ if and only if $ w_j\in\varphi\left(s_i\right) $.
\end{Defn}

We next define the \textit{blocking coalition}, which is a crucial factor that can make the matching unstable.

\begin{Defn}(Blocking coalition)
	Under a given matching $ \varphi $, worker $ w_j $ and task vector $ \mathbb{S} \subseteq \bm{S}$ may form one of the following two types of blocking coalition, denoted by $ \left(w_j; \mathbb{S}\right) $.
	
\noindent \textbf{Type 1 blocking coalition:} Type 1 blocking coalition satisfies the following two conditions:
	
	\noindent
	$\bullet$ Worker $ w_j $ prefers a task vector $ \mathbb{S} \subseteq \bm{S} $ rather than its currently matched task vector $ \varphi(w_j) $, i.e., 
	\begin{equation}\label{key}\tag{18}
		\begin{aligned}
				\overline{U^W}(w_j,\mathbb{S})>\overline{U^W}(w_j,\varphi(w_j)). 
		\end{aligned} 
	\end{equation}

	\noindent
	$\bullet$ Every task in $ \mathbb{S} $ prefers to recruit workers rather than being matched to its currently matched/assigned worker set. That is, for any task $ s_i\in \mathbb{S} $, there exists a worker set $ \varphi^\prime(s_i) $ that constitutes the workers that need to be evicted, satisfying
	\begin{equation}\label{key}\tag{19}
			\begin{aligned}
				\overline{U^{S}}\left(s_i,\left\{\varphi\left( s_i \right)\backslash\varphi^{\prime}\left( s_i \right)\right\} \cup \left\{ w_{j} \right\} \right) > \overline{U^{S}}\left(s_i,\varphi\left( s_i \right) \right).\\
		\end{aligned}
	\end{equation} 
	
\noindent \textbf{Type 2 blocking coalition:} Type 2 blocking coalition satisfies the following two conditions:
	
	\noindent
	$\bullet$ Worker $ w_j $ prefers executing task vector $ \mathbb{S} \subseteq \bm{S} $ to its currently matched task vector $ \varphi(w_j) $, i.e.,
	\begin{equation}\label{key}\tag{20}
			\begin{aligned}
				\overline{U^W}(w_j,\mathbb{S} )>\overline{U^W}(w_j,\varphi(w_j) ).
		\end{aligned}
	\end{equation} 
	
	\noindent
	$\bullet$ Every task in $ \mathbb{S} $ prefers to further recruit worker $ w_j $ in conjunction to its currently matched/assigned worker set\footnote{Note that the Type 2 blocking coalition is a special case of the Type 1 blocking coalition. When no task $s_i \in \mathbb{S}$ needs to evict any of its currently matched workers (i.e., $\varphi'(s_i) = \emptyset$), the Type 1 coalition structurally degenerates into a Type 2 coalition.}. That is, for any task $ s_i\in \mathbb{S} $, we have
	\begin{equation}\label{key}\tag{21}
			\begin{aligned}
				\overline{U^{S}}(s_i,\varphi(s_i)\cup\left\{ w_{j} \right\})>\overline{U^{S}}(s_i,\varphi(s_i) ) .
		\end{aligned} 
	\end{equation}
\end{Defn}

Recall the above definitions, the Type 1 blocking coalition can lead to the unstability of matching since a task is incentivized to recruit another set of workers that can bring it with a higher expected utility. Similarly, the Type 2 blocking coalition makes the matching unstable since the task has a left-over budget to recruit more workers, which can further helps with increasing its expected utility.
\subsection{Problem Formulation}
We conduct the service provisioning in futures trading stage as a M2M matching, occurring prior to practical transactions. Our goal is to achieve stable mappings between tasks and workers to facilitate long-term contracts. Accordingly, each task $ s_i \in \bm{S} $ aims to maximize its overall expected utility, which is mathematically formulated as the following optimization problem:
\begin{subequations}
	\begin{align}
		\bm{\mathcal{F}^S}:~&\underset{{\varphi\left(s_i\right)}}{\max}~\overline{U^S}\left(s_i,\varphi\left(s_i\right)\right)\tag{22}\\
		\text{s.t.}~~~
		&\varphi\left(s_i\right)\subseteq\bm{W} \tag{22a}\\
		&\sum_{w_j\in\varphi(s_i)} p^F_{i,j}\le B_i\tag{22b}\\
		&R^{S}_1\left( s_{i},\varphi \left( s_{i} \right) \right)\leq \rho_1,\tag{22c}
	\end{align}
\end{subequations}
where $ 0<\rho_1\le 1 $ represents a risk threshold\footnote{The threshold is not a fixed parameter of the mechanism, but is a configurable variable set by individual workers or task owners, allowing flexible adaptation to different risk preferences or deployment scenarios.}. In $ \bm{\mathcal{F}^S} $, constraint (22a) forces the recruited workers $ \varphi(s_i) $ within set $ \bm{W} $, constraint (22b) ensures that the expenses of task $ s_i $ devoted to recruiting workers $ \varphi(s_i) $ does not exceed its budget $ B_i $, and constraint (22c) dictates the tolerance of each task on receiving an undesired utility, with its derivation detailed in Appx. B.2.
Furthermore, each worker $ w_j\in\bm{W} $ aims to maximize its expected utility, as modeled by
\begin{subequations}
	\begin{align}
		\bm{\mathcal{F}^W}:~&\underset{{\varphi\left(w_j\right)}}{\max}~\overline{U^W}\left(w_j,\varphi\left(w_j\right)\right)\tag{23}\\
		\text{s.t.}~~~
		&\varphi\left(w_j\right)\subseteq\bm{S} \tag{23a}\\
		&c_{i,j}\le p^F_{i,j},~\forall s_i\in \varphi \left( w_j \right) \tag{23b}\\
		&R^{W}_1\left( w_j,s_i \right)\leq \rho_2,~\forall s_i\in \varphi \left( w_j \right)\tag{23c}\\
		&R^{W}_2\left( w_j,s_i \right)\leq \rho_3,~\forall s_i\in \varphi \left( w_j \right),\tag{23d}
	\end{align}
\end{subequations}
where $ \rho_2 $ and $\rho_3 $ are risk thresholds falling in interval $ (0, 1] $. In $ \bm{\mathcal{F}^W} $, constraint (23a) ensures that task set $ \varphi(w_j) $ belongs to set $ \bm{S} $, and constraint (23b) shows that the payments asked by $ w_j $ can cover the corresponding service costs; constraints (23c) and (23d) manages the possible risks that each worker may confront during practical transactions, with their derivation detailed in Appx. B.1.

Our proposed futures trading stage aims to solve a multi-objective optimization (MOO) problem that involves (22) and (23), which is challenging to be solved. Moreover, when confronting uncertain MCS networks and potential risks, conventional M2M matching can struggle to directly address these challenges. After extensive investigations, we explore an improved heuristic method to determine workers' preference lists and the paths for task execution, while leveraging a futures-based M2M stable matching algorithm to address the MOO problem, as detailed in the following section.

\begin{figure}[t]{
	\centering
	\includegraphics[width=0.9\columnwidth]{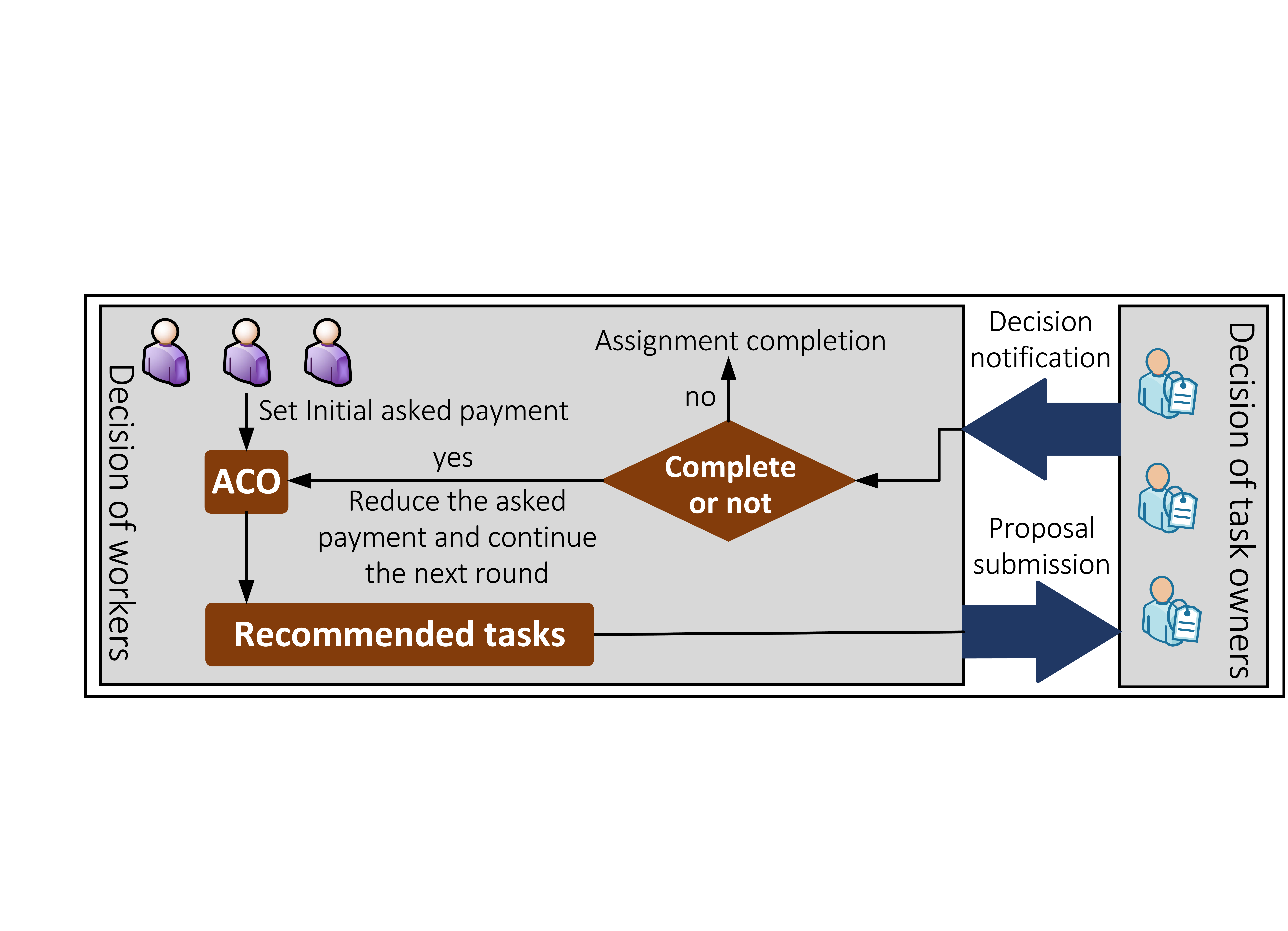}
	\caption{A flow chart, showing the distributed task-worker interaction model regarding our proposed FT-SMP$^3$.}}
\end{figure}

\subsection{Solution Design}
\begin{figure*}[b!] 
	\centering
	\vspace{-0.2cm}
	\hrulefill
	\begin{equation}\tag{24}
		pr^\mathbbm{k}_{m,n}=\left\{
		\begin{aligned}
			&\frac{\tau^{\varepsilon_1}_{m,n}\eta^{\varepsilon_2}_{m,n}(1/width_n)^{\varepsilon_3}(1/wait_n)^{\varepsilon_4}}{\sum_{k\in J_\mathbbm{k}(m)}\tau^{\varepsilon_1}_{m,n}\eta^{\varepsilon_2}_{m,n}(1/width_n)^{\varepsilon_3}(1/wait_n)^{\varepsilon_4}},~~~\text{if } n \in J_\mathbbm{k}(m)\\
			&0,~~~~~~~~~~~~~~~~~~~~~~~~~~~~~~~~~~~~~~~~~~~~~~~~~~~~~~~~~~~~~~~~~~~~~~~~~~\text{otherwise}
		\end{aligned} \right.
	\end{equation}
\end{figure*}

We next develop the \textbf{f}utures \textbf{t}rading-driven \textbf{s}table \textbf{m}atching and \textbf{p}re-\textbf{p}ath-\textbf{p}lanning (FT-SMP$^3$) to \textit{achieve efficient and stable matching between workers and spatial-temporal tasks, thereby obtaining \textit{i)} a set of long-term contracts among them to guide future transactions, and \textit{ii)} the path that each worker may follow.} Specifically, FT-SMP$^3$ consists of two key aspects: \textit{i)} \textbf{e}nhanced \textbf{a}nt \textbf{c}olony \textbf{o}ptimization for \textbf{p}ath \textbf{p}re-\textbf{p}lanning and \textbf{t}ask \textbf{r}ecommendation (EACO-P$^3$TR), taking into account diverse task demands, payments that workers can receive, and costs for performing tasks, which aims to recommend tasks to workers along with corresponding pre-planned paths; and \textit{ii)} \textbf{f}utures \textbf{t}rading-driven \textbf{m}any-to-\textbf{m}any (FT-M2M) matching, enabling long-term contracts between workers and tasks to achieve stable services.

Fig. 2 shows the flow chart of our proposed FT-SMP$^3$, where the interaction between task owners and workers can be realized through deploying multiple rounds. Key steps are detailed below: 

\noindent $\bullet$ At the beginning of each round, each worker utilizes our proposed EACO-P$^3$TR algorithm (detailied in Sec. 4.4.1) to comprehensively evaluate diverse task demands and asked payments, thereby selecting an optimal vector of tasks that maximizes its expected utility (i.e., (23)). This selected task vector represents the worker's tasks of interest for the current round. The worker then announces its information to the corresponding interested tasks.

\noindent
$\bullet$ After collecting workers' reports, each task temporarily selects some workers who can offer the highest expected service quality as candidates, under its budget constraints (i.e., (22b)), and then informs all the workers.

\noindent
$\bullet$ Following the decisions of tasks, if a worker is rejected by a task, it can reduce its asked payments to enhance its competitiveness in the next round; while if it is selected, its asked payment remains unchanged in the next round.

\noindent
$\bullet$ Repeat the above steps until either of the following conditions is met: \textit{i)} all workers are successfully recruited by their interested tasks; or \textit{ii)} no worker can further reduce its asked payment (e.g., to avoid negative utility and mitigate risks). Once these conditions are met, the corresponding matching result is obtained, which includes information on long-term contracts. Furthermore, the final round's vector of interested tasks for each worker represents their practical task execution order in future transactions, thereby determining the pre-planned paths for workers.

\subsubsection{Proposed EACO-P$^3$TR Algorithm}
Inspired by the ant colony optimization (ACO)\cite{ACO 1} algorithm upon considering diverse characteristics of tasks (e.g., varying time windows and locations) and workers (e.g., various costs), we develop the EACO-P$^3$TR algorithm. This algorithm recommends proper tasks to each worker along with a pre-planned path (e.g., the sequence of tasks), to \textit{maximize workers' expected utility}. Regarding a worker $w_j$ and all the tasks as vertices and the edges between them, we construct a directed complete graph denoted by $\bm{G_j}(\bm{V_j}, \bm{E_j})$. In this graph, $\bm{V_j} = \left\{v_0, v_1, \ldots, v_{|\bm{S}|}\right\}$ represents the set of vertices ($v_0$ denotes the worker $w_j$, $v_1,..., v_{|\bm{S}|}$ represent tasks), indexed by $m$ and $n$ for convenience of analysis, and $\bm{E_j} = \left\{(v_m, v_n) | v_m, v_n \in \bm{V_j}, m \neq n\right\}$ comprises the edge set. Note that the location of $v_0$ is denoted by $l^{\text{w}}_j\langle 0\rangle$, and the locations of $v_1, v_2, \ldots, v_{|\bm{S}|}$ are $l^{\text{s}}_1, l^{\text{s}}_2, \ldots, l^{\text{s}}_{|\bm{S}|}$.

In EACO-P$^3$TR, we use a set $\bm{\mathbb{K}}$ to denote the collection of ants. Each ant $\mathbbm{k}\in \mathbb{K}$ simulates the moving trajectory of worker $w_j$, namely, starting from $v_0$. {It considers comprehensive information such as the current timeslot, the time window of tasks, the location of vertices (location of tasks), the expected timeslot required to travel to each task, and the expected utility upon reaching each task, while then selects a feasible vertice to visit next (i.e., the system evaluates whether the expected travel time from the worker's current location to the target vertex, when added to the current timeslot, will bring success arrival before the end of the target vertex's time window).} This process continues until no visitable applicable vertices left. The state transition rule of EACO-P$^3$TR, i.e., the probability of ant $\mathbbm{k}$ moving from task $v_m$ to $v_n$ ($m\neq n\neq 0$), can be defined as (24), where $\varepsilon_1$, $\varepsilon_2$, $\varepsilon_3$, and $\varepsilon_4$ are positive weight coefficients. Besides, $\tau_{m,n}$ represents the pheromone deposited for the transition from $v_m$ to $v_n$, as given by (25), while $\eta_{m,n}$ denotes the distance between $v_m$ and $v_n$, and $width_n$ denotes the width of the time window for task $s_i$ in vertice $v_n$ (i.e., $width_n = t_i^{\text{e}} - t_i^{\text{b}}$). Moreover, $wait_n$ represents the waiting time if the ant arrives early (e.g., before the task starts), where $wait_n=\max\left\{t_n^{\text{b}}-(t_{m,n}^{\text{ini}}+\tau_{m,n}^{\text{move}}+\tau_{m,n}^{\text{D}}),0\right\}$. Moreover, $J_\mathbbm{k}(m)$ is the set of remaining feasible tasks that ant $\mathbbm{k}$ can complete at vertice $v_m$, and the tasks $s_i \in J_\mathbbm{k}(m)$ should satisfy constraints (23c) and (23d).

When a path has been found by the ant colony, pheromones along with the edges are updated as
\begin{equation}\tag{25}
	\tau_{m,n}\leftarrow(1-\theta)\tau_{m,n}+\mathrm{\Delta}\tau_{m,n},
\end{equation}
where $0 < \theta < 1$ is the pheromone evaporation coefficient, and $\mathrm{\Delta}\tau_{m,n}$ is the pheromone deposited by ants, which can be calculated as
\begin{equation}\tag{26}
	\mathrm{\Delta}\tau_{m,n}=\left\{
	\begin{aligned}
		&\sum_{(v_m,v_n)\in\bm{M^{\text{max}}}}\triangle
		\overline{U^W_{m,n}},~\text{if edge }(v_m,v_n) \in \bm{M^{\text{max}}}\\
		&0,~~~~~~~~~~~~~~~~~~~~~~~~~~~~~~~~~~\text{otherwise},
	\end{aligned} \right.
\end{equation}
where $\triangle\overline{
	U^W_{m,n}}$ denotes the expected utility for worker $w_j$ from vertice $v_m$ to finish the task in vertice $v_n$. $\bm{M^{\text{max}}}$ is the path (e.g., a vector of the completed tasks of $w_j$) with the maximum utility obtained from all ants, i.e., $\bm{M^{\text{max}}}= \left\{\bm{M^\mathbbm{k}}|\max~\sum_{(m,n)\in\bm{M^\mathbbm{k}}}\triangle\overline{U^W_{m,n}},~\forall \mathbbm{k}\in\mathbb{K}\right\}$.

\begin{algorithm}[t] 
	{\small
		\caption{{Proposed EACO-P$^3$TR Algorithm}}
		\LinesNumbered 
		\textbf{Input:} worker $w_j$, task set $\bm{S}$\ 
		
		\textbf{Initialization:} $\tau_{m,n}\leftarrow \tau_0$, $iter=1$, $v_m=v_0$, $iter_{\text{max}}$, $\mathbb{K}$
		
		\While{$iter\le iter_{\text{max}}$}{
			\For{$\forall k\in \mathbb{K}$}{
				$\bm{M^\mathbbm{k}}=\varnothing$
				
				\While{not end condition}{\textbf{Calculate:} $J_\mathbbm{k}(m)$, $pr^\mathbbm{k}_{m,n}$
					
					\If{$\forall J_\mathbbm{k}(m) \neq \varnothing$}{$n \leftarrow$ ant $\mathbbm{k}$ randomly selects the next task $n$ from $J_\mathbbm{k}(m)$ according to $pr^\mathbbm{k}_{m,n}$
						
						$\bm{M^\mathbbm{k}}\leftarrow \left\{\text{edge}(m,n)\right\}\cup\bm{M^\mathbbm{k}}$
						
						$m\leftarrow n$
					}
					\Else{break loop}
				}
			}
			
			$\bm{M^{\text{max}}}\leftarrow \left\{\bm{M^\mathbbm{k}}|\max~\sum_{(m,n)\in\bm{M^\mathbbm{k}}}\triangle\overline{U^W_{m,n}},~\forall k\in\bm{K}\right\}$
			
			update $\tau_{m,n}$

			$iter=iter+1$
		}
		$L^{\text{best}}_j \leftarrow $ selects the task vector from $ \bm{M^{\text{max}}} $
		
		\textbf{Output:} $L^{\text{best}}_j$}
\end{algorithm}

Alg. 1 details the pseudo-code of our designed EACO-P$^3$TR algorithm, including several key steps: 

\noindent
\textbf{Step 1. Initialization:} According to the information associated with workers and tasks (e.g., their locations, time windows of tasks, payments of workers, and the expected time and energy consumed on performing tasks), Alg. 1 starts with constructing a directed complete graph $ \bm{G_j}(\bm{V_j}, \bm{E_j}) $ for each worker $w_j$, denoted by $v_0$.

\noindent
\textbf{Step 2. Path exploration:} In each iteration, ant $\mathbbm{k}$ at each vertice evaluates multiple factors involving: the present time, the expected time traveling to the next unvisited vertice, the expected time to perform the task, and the time window. This evaluation is conducted to ascertain whether the task can be completed within the specified time frame, thereby obtaining $J_\mathbbm{k}(m)$ and $pr^\mathbbm{k}_{m,n}$ (line 7, Alg. 1). Then, the next vertice (task) is chosen from set $J_\mathbbm{k}(m)$ according to (24). This iteration proceeds until no further executable vertices (line 8, Alg. 1).

\noindent
\textbf{Step 3. Pheromone update:} Subsequently, the expected utilities $\sum_{(m,n)\in\bm{M^\mathbbm{k}}}\triangle
\overline{U^W_{m,n}}$ associated with different vertice sequence $\bm{M^\mathbbm{k}}$ are compared, among which, the path with the maximum expected utility (i.e., $\bm{M^{\text{max}}}$) is selected for pheromone updating (lines 14-15, Alg. 1), following the procedures outlined in (25) and (26).

\noindent
\textbf{Step 4. Optimal path output:} After finishing all the iterations, the path $\bm{M^{\text{max}}}$ with the highest expected utility can be determined as the pre-planned path for worker $w_j$. Accordingly, the tasks along this path are recommended.

\subsubsection{Proposed FT-M2M matching algorithm}
\begin{algorithm}[t!] 
		{\small\caption{{Proposed FT-M2M Matching Algorithm}}
		\LinesNumbered 
		\textbf{Initialization:} $ k \leftarrow 1 $, $ p_{i,j}\left\langle 1 \right\rangle \leftarrow p^{\text{Desire}}_{i,j}$, for $ \forall i,j $, $ {flag}_{j} \leftarrow 1 $, $\mathbb{Y}\left( w_{j} \right)\leftarrow \varnothing$, $\mathbb{Y}\left( s_{i} \right)\leftarrow \varnothing$\ 

		\While{$ {flag}_{j} $}{
			\textbf{$ {flag}_{j} \leftarrow 0 $}
			
			$ L^{\text{best}}_j \leftarrow \text{Alg. 1}$, $ \mathbb{Y}\left( w_{j} \right) \leftarrow L^{\text{best}}_j $
			
			\If{$ \forall\mathbb{Y}\left( w_{j} \right) \neq \varnothing $}{
				\For{$ \forall s_i \in \mathbb{Y}\left( w_{j} \right) $}{$ w_j $ sends a proposal about its information to $ s_i $}
				
				\While{
					$ \Sigma_{w_{j}\in \bm{W}}{flag}_{j} > 0 $}{
					Collect proposals from the workers in $ \bm{W} $, e.g., using $ {\widetilde{\mathbb{Y}}}\left(s_i\right) $ to include the workers that send proposals to $ s_i $
					
					$ \mathbb{Y}(s_i)\leftarrow $ choose workers from $ {\widetilde{\mathbb{Y}}}\left(s_i\right) $ that can achieve the maximization of the expected utility by using DP under budget $ B_i $
					
					$ s_i $ temporally accepts the workers in $ \mathbb{Y}(s_i) $, and rejects the others
				}
				
				\For{
					$ \forall s_i \in \mathbb{Y}\left( w_{j} \right) $
				}{
					\If{$ w_j $ is rejected by $ s_i $, $ p_{i,j}\left\langle k \right\rangle>c_{i,j} $, and risk constraint (23c) is met}{
						$ p_{i,j}\left\langle {k + 1} \right\rangle \leftarrow \max\left\{ p_{i,j}\left\langle k \right\rangle - \mathrm{\Delta}p~,{ c}_{i,j} \right\} $}
					\Else{$ p_{i,j}\left\langle {k + 1} \right\rangle \leftarrow p_{i,j}\left\langle k \right\rangle $}
				}
				
				\If{there exists $p_{i,j}\left\langle {k + 1} \right\rangle \neq p_{i,j}\left\langle k \right\rangle\ $, $ \forall s_i\in\mathbb{Y}\left( w_{j} \right) $}{
					$ {flag}_j\leftarrow 1 $,	$ k\leftarrow k+1 $\
				}
			}
		}

		\If{
			$ R^{S}\left( s_i ,\varphi\left( s_i \right) \right) \le \rho_1 $}{$ s_i $ gives up trading in the futures trading stage}
			
			$\varphi(s_i)\leftarrow\mathbb{Y}(s_i)$, $\varphi(w_j)\leftarrow \mathbb{Y}(w_j)$
			
		\textbf{Return:} $\varphi(s_i)$, $\varphi(w_j)$, $L_j^{\text{best}}$.}
\end{algorithm}
Following the pre-planned paths, we then introduce FT-M2M matching (see Alg. 2), which facilitates stable, efficient, and risk-aware matching between tasks and workers, as well as their long-term contracts.

\noindent
\textbf{Step 1. Initialization:} At the beginning of Alg. 2, each worker $ w_j $'s asked payment is set by $ p_{i,j}\left\langle 1 \right\rangle = p^{\text{Desire}}_{i,j} $ (line 1, Alg. 2), $ \mathbb{Y} (w_j ) $ contains the interested tasks of $ w_j $ and $ \mathbb{Y}(s_i) $ involves the workers temporarily selected by $ s_i $.

\noindent
\textbf{Step 2. Proposal of workers:} At each round $ k $, each worker $ w_j $ first chooses tasks from $L^{\text{best}}_j$ (the optimal task vector for worker $w_j$ according to its expected utility, derived from Alg. 1), and records them in $ \mathbb{Y}(w_j ) $\footnote{Note that workers are not rigidly bound to their initially preferred task. Instead, at the beginning of each round, every worker re-runs the EACO-P$^3$TR module to dynamically update its task-path preferences based on current expected utilities. As a result, if the utility of the previous best-matched task drops (e.g., due to repeated price reductions), the worker may switch to a new task with higher expected utility. This allows workers to avoid sticking with outdated preferences and promotes adaptive decision-making.}
. Then, $ w_j $ sends a proposal to each task in $ \mathbb{Y}(w_j) $, including its asked payments $ p_{i,j}\left\langle k \right\rangle $, probability of completing $s_i$ (i.e., $\beta_{i,j}$), and expected service quality of sensing data $\text{E}[Q_{i,j}]$ (line 7, Alg. 2).

\noindent
\textbf{Step 3. Worker selection on tasks' side:} After collecting the information from workers in set $ {\widetilde{\mathbb{Y}}}\left(s_i\right) $, each task $ s_i $ solves a 0-1 knapsack problem,
which can generally be solved via dynamic programming (DP) \cite{RWork_stable 5,DP1,DP2} (line 10, Alg. 2), determine a collection of temporary workers (e.g., set $ \mathbb{Y} (s_i)$), where $\mathbb{Y} (s_i)\subseteq {\widetilde{\mathbb{Y}}}\left(s_i\right) $ that can bring the maximum expected utility, under budget $ B_i $. Then, each $ s_i $ reports its decision on worker selection during the current round to workers. 

\noindent
\textbf{Step 4. Decision-making on workers' side:} After obtaining the decisions from each task $ s_i\in\mathbb{Y} (w_j) $, worker $ w_j $ inspects the following conditions:

\noindent
$\bullet$ \textbf{Condition 1.} If $ w_j $ is accepted by $ s_i $ or its current asked payment $ p_{i,j}\left\langle k \right\rangle $ equals to its cost $ {\ c}_{i,j} $ or risk constraint (23c) isn't met, its payment remains unchanged (line 16, Alg. 2);

\noindent
$\bullet$ \textbf{Condition 2.} If $ w_j $ is rejected by a task $ s_i $, its asked payment $ p_{i,j}\left\langle k \right\rangle $ can still cover its cost $ {c}_{i,j} $ and risk constraint (23c) is met, it decreases its asked payment\footnote{The asked payment decrement $\triangle p$ governs the trade-off between convergence speed and solution optimality. A larger $\triangle p$ accelerates convergence by updating payments more aggressively, but may miss better matching opportunities. In contrast, a smaller $\triangle p$ enables more fine-grained utility balancing at the cost of additional matching rounds. Please see Appx. G.3 for a detailed theoretical discussions.} for $ s_i $ in the next round, while avoiding a negative utility (line 14, Alg. 2):
\begin{equation}\label{key}	\tag{27}
	{
		\begin{aligned}
			p_{i,j}\left\langle {k + 1} \right\rangle = \max\left\{ p_{i,j}\left\langle k \right\rangle - \mathrm{\Delta}p_j~,{~c}_{i,j} \right\}.
	\end{aligned} }
\end{equation}

\noindent
\textbf{Step 5. Repeat:} If all the asked payments stay unchanged from $ (k-1)^{\text{th}} $ round to $ k^{\text{th}} $ round, the matching terminates at round $ k $. We use $ \Sigma_{w_j\in \bm{W}}{flag}_j=0 $ to denote this situation (line 3, Alg. 2).
Otherwise, the algorithm repeats the above steps (e.g., lines 2-18, Alg. 2) in the next round.

\noindent
\textbf{Step 6. Risk analysis:} When the algorithm is terminated, each task $ s_i $ will choose whether to sign long-term contracts with the matched workers according to its risk estimation (constraint (22c)).

{Note that the computational complexity of our proposed FT-SMP$^3$ can be described by $\mathcal{O}\left(\mathcal{X} \cdot \left(|\bm{W}| \cdot iter_{\text{max}} \cdot \mathbb{K} \cdot |\bm{S}|^2 + \sum_{s_i \in \bm{S}} |{\widetilde{\mathbb{Y}}}(s_i)| \cdot B_i \right)\right)$, where $\mathcal{X}$ denotes the number of matching rounds (more details can be found in Appendix C)}.

As our FT-M2M matching is deployed before future practical transactions, our focus lies in distinctive objectives, e.g., the expectation of workers' and tasks' utilities, as well as control of potential risk, which greatly differs us from conventional matching mechanisms. Hereafter, we will cover key properties on designing our unique matching in futures trading stage.

\begin{Defn}(Individual rationality of FT-M2M matching) For both tasks and workers, a matching $ \varphi $ is individually rational when the following conditions are satisfied:
	
	\noindent
	$\bullet$ For tasks: \textit{i)} the overall payment of a task $s_i$ paid to matched workers $\varphi\left(s_i\right)$ does not exceed $B_i$, i.e., constraint (22b) is met; \textit{ii)} the risk of each task receiving an undesired service quality is controlled within a certain range, i.e., constraint (22c) is satisfied.
	
	\noindent
	$\bullet$ For workers: \textit{i)} the risk of each worker receiving an undesired expected utility is controlled within a certain range, i.e., constraint (23c) is satisfied; \textit{ii)} the risk of each worker failing to complete matched tasks is controlled within a certain range, i.e., constraint (23d) is satisfied.
\end{Defn}

\begin{Defn}(Fairness of FT-M2M matching): A matching $\varphi$ is fair if and only if it does not impose type 1 blocking coalition.\end{Defn}
\begin{Defn}(Non-wastefulness of FT-M2M matching): A matching $\varphi$ is non-wasteful if and only if it does not impose type 2 blocking coalition. \end{Defn}

\begin{Defn}(Strong stability of FT-M2M matching) The proposed matching is strongly stable if it is individual rational, fair, and non-wasteful.
\end{Defn}

Note that competitive equilibrium represents a conventional concept in economic behaviors, playing an important role in analyzing the performance of commodity markets upon having flexible prices and multiple players. When the considered market arrives at the competitive equilibrium, there exists a price at which the number of task owners that will pay is equal to the number of workers that will sell\cite{pareto 1}. Correspondingly, the competitive equilibrium of FT-M2M matching is defined below.

\begin{Defn}(Competitive equilibrium associated with trading between workers and task owners in futures trading stage) The trading between workers and task owners reaches a competitive equilibrium if the following conditions are satisfied:
	
	\noindent
	$\bullet$ {For each worker $ w_j \in \bm{W} $, if $ w_j $ is associated with a task owner $ s_i\in \bm{S} $, then $ \text{E}[c_{i,j}]\leq p^F_{i,j} $~;}
	
	\noindent
	$\bullet$ For each task $ s_i\in \bm{S} $, $ s_i $ is willing to trade with the worker that can bring it with the maximum expected utility;
	
	\noindent
	$\bullet$ For each task \( s_i \) in set \( \bm{S} \), when \( s_i \) does not recruit more workers, it indicates that the remaining budget after deducting the payments made to matched workers is insufficient to recruit an additional worker.
\end{Defn}

For a multi-objective optimization problem (e.g., $ \bm{\mathcal{F}^W} $ and $ \bm{\mathcal{F}^S} $), a Pareto improvement occurs when the \textit{expected social welfare (referring to a summation of expected
	utilities of workers and task owners in our considered market)} can be increased with another feasible matching result\cite{pareto 1}. Thus, a matching is weak Pareto optimal when there is no Pareto improvement.

\begin{Defn}(Weak Pareto optimality of trading between tasks and workers in futures trading stage) The proposed matching is weak Pareto optimal if there is no Pareto improvement.
\end{Defn}
We show that our proposed FT-M2M matching of FT-SMP$^3$ can support the above-discussed properties, while the corresponding analysis and proofs are given by Appx. D, due to space limitation.

\section{Proposed ST-DP$^2$WR Mechanism}
When comes to practical transactions, contractual workers and task owners can follow the long-term contracts and pre-determined paths directly. Nevertheless, the dynamic and uncertain nature of MCS networks may stop workers to complete assigned tasks on the pre-planed path. To further enhance task completion and utilities for both parties, we propose ST-DP$^2$WR mechanism as a quick and reliable backup, including \textit{i)} for workers, we investigate a \textbf{d}eep \textbf{q}-\textbf{n}etwork with \textbf{p}rioritized experience replay-enhanced \textbf{d}ynamic \textbf{r}oute \textbf{o}ptimization (DQNP-DRO) algorithm. This algorithm assists workers to decide whether to follow the pre-planned path, or to abandon a contractual task (e.g., a worker can not get to a task in time due to delay events), or to select a new task (which has not been listed in the pre-planned path). And \textit{ii)} for tasks with unsatisfying service quality (some workers may be unable to serve them due to uncertain delay events), we introduce the \textbf{s}pot trading-based \textbf{t}emporal \textbf{m}any-to-\textbf{m}any (ST-M2M) matching algorithm, helping recruit temporary workers. To avoid confusion with futures trading, we rewrite $p_{i,j}$, $q_{i,j}$, and $\beta_{i,j}$ in the spot trading stage as $p^S_{i,j}$, $q^S_{i,j}$, and $\beta_{i,j}^S$, respectively.

To facilitate analysis, we define several key notations:

\noindent
$\bullet$ $\bm{W^\prime}\langle t\rangle$: the set of workers that can serve temporary tasks at timeslot $t$, i.e., $ \bm{W^\prime}\langle t\rangle\in\bm{W}$;

\noindent
$\bullet$ $\bm{S^\prime}\langle t\rangle$: the set of sensing tasks that have surplus budget to recruit temporary workers at timeslot $t$;

\noindent
$\bullet$ $\nu_t(s_i)$: the set of workers temporary recruited for processing task $ s_i $ at timeslot $t$, i.e., $ \nu_t(s_i) \in \bm{W^\prime}\langle t\rangle $;

\noindent
$\bullet$ $ \nu_t(w_j) $: the set of sensing tasks temporary assigned to worker $ w_j $ at timeslot $t$, i.e., $ \nu_t(w_j) \in \bm{S^\prime}\langle t\rangle $.
\subsection{Utility of tasks and workers}
Utility of task $s_i\in \bm{S^\prime}\langle t\rangle$ consists of two parts: \textit{i)} obtained service quality, and \textit{ii)} remaining budget, as given by
\begin{equation}\label{key}\tag{28}
	\begin{aligned}
		&U^{S\prime}_t (s_i,\nu_t(s_i))=\mathbb{V}_2\sum_{w_j\in \nu_t(s_i)} \frac{1}{AGE_{i,j}}+\\& \mathbb{V}_3 \left(-\sum_{w_j\in \nu_t(s_i )}\left((1-\beta_{i,j}^{S})q_{i,j}^{S}-\beta_{i,j}^{S}p_{i,j}^{S}\right)\right).
	\end{aligned}
\end{equation}

We also define the utility of worker $w_j \in \bm{W^\prime}$ to encompass three key parts: \textit{i)} the overall payment minus the cost of completing tasks, \textit{ii)} the service cost for failing to complete tasks, and \textit{iii)} the penalty incurred by worker $w_j$ for failing to complete tasks. This is given by:
\begin{equation}\label{key}\tag{29}
	\begin{aligned}
		&U_t^{W\prime} (w_j,\nu_t(w_j))=\sum_{s_i\in\nu_t(w_j)}\beta_{i,j}^{S}\left(p_{i,j}^{S}-c_{i,j}\langle t_{i,j}^{\text{ini}} \rangle\right)\\&-\sum_{s_i\in\nu_t(w_j)}(1-\beta_{i,j}^{S})\left(c_{i,j}^{\text{part}}\langle t_{i,j}^{\text{ini}} \rangle+q_{i,j}^{S}\right).
	\end{aligned}
\end{equation}
Due to the dynamic nature of the MCS networks, workers during this stage can still confront risks:

\textit{i) The risk of receiving an unsatisfying utility:} Each worker $w_j \in \bm{W}$ serving task $s_i \in \nu_t(w_j)$ faces the risk of obtaining an unsatisfying utility, given by:
\begin{equation}\label{key}\tag{30}
	R^{W\prime}_1 (w_j ,s_i)=\text{Pr}\left(U^{W\prime}_t (w_j,s_i)<u_{\text{min}}\right),~\forall s_i\in \nu_t(w_j)
\end{equation}

\textit{ii) The risk on failing to complete the task:} Each worker $w_j$ may not have enough time to complete task $s_i \in \nu_t(w_j)$ due to delay events, as defined by
\begin{equation}\label{key}\tag{31}
	R_2^{W\prime} (w_j ,s_i)=\text{Pr}\left(\beta_{i,j}^S=0\right),~~~\forall s_i\in \nu_t(w_j).
\end{equation}
Besides, the aforementioned risks should be managed within an acceptable range; otherwise, worker $w_j$ will not service task $s_i$.
\subsection{Problem Formulation}
In the spot trading stage, each task $ s_i \in \bm{S^\prime}\langle t\rangle $ aims to maximize its overall utility, which can be mathematically formulated as the following optimization problem:
\begin{subequations}
	\begin{align}
		\bm{\mathcal{F}^{S\prime}}:~&\underset{{\nu_t\left(s_i\right)}}{\max}~U_t^{S\prime}\left(s_i,\nu_t\left(s_i\right)\right)\tag{32}\\
		\text{s.t.}~~~
		&\nu_t\left(s_i\right)\subseteq\bm{W^\prime}\langle t\rangle \tag{32a}\\
		&\sum_{w_j\in\nu_t(s_i)} p^s_{i,j}\le B_i^t,\tag{32b}
	\end{align}
\end{subequations}
 where \( B_i^t \) represents the available budget of task \( s_i \) at timeslot \( t \), also includes the compensation from workers for can not complete \( s_i \) prior to timeslot \( t \) (as stipulated in long-term contracts). In $ \bm{\mathcal{F}^{S\prime}} $, constraint (32a) forces recruited workers $ \nu_t(s_i) $ to belong to set $ \bm{W^\prime}\langle t\rangle $, constraint (32b) ensures that the expense of task $ s_i $ devoted to recruiting workers $ \nu_t(s_i) $ does not exceed the its remaining budget $ B_i^t $. Besides, each worker $ w_j\in\bm{W^\prime} $ aims to maximize its utility, as follow
\begin{subequations}
	\begin{align}
		\bm{\mathcal{F}^{W\prime}}:~&\underset{{\nu_t\left(w_j\right)}}{\max}~{U_t^{W\prime}}\left(w_j,\nu_t\left(w_j\right)\right)\tag{33}\\
		\text{s.t.}~~~
		&\nu_t\left(w_j\right)\subseteq\bm{S^\prime}\langle t\rangle \tag{33a}\\
		&c_{i,j}\langle t^{\text{ini}} \rangle\le p^s_{i,j},~\forall s_i\in \nu_t \left( w_j \right) \tag{33b}\\
		&R^{W^\prime}_1\left( w_j,s_i \right)\leq \rho_4,~\forall s_i\in \nu_t \left( w_j \right)\tag{33c}\\
		&R^{W^\prime}_2\left( w_j,s_i \right)\leq \rho_5,~\forall s_i\in \nu_t \left( w_j \right),\tag{33d}
	\end{align}
\end{subequations}
where $ \rho_4 $ and $\rho_5 $ are risk thresholds fall in interval $ (0, 1] $. In $ \bm{\mathcal{F}^{W\prime}} $, constraint (33a) ensures that task set $ \nu_t(w_j) $ belongs to set $ \bm{S^\prime}\langle t\rangle $, and constraint (33b) ensures the payments asked by $ w_j $ can cover their service costs; constraints (33c) and (33d) control possible risks, and their derivations are detailed by Appx. E. 

The ST-M2M matching algorithm focuses on a MOO problem. Our goal is \textit{to achieve stable and efficient matching between tasks that require temporary recruitment and workers during a specific timeslot in spot trading stage, thereby enhancing the service quality obtained by the tasks through temporary recruitment}. 

\subsection{Solution Design}
Facing dynamic networks and uncertainties factors (e.g., delay events), workers need to make decisions through our proposed DQNP-DRO algorithm (i.e., whether to abandon the current task or continue with it) to maximize their utility. For tasks that are abandoned, recruiting temporary workers through ST-M2M matching can be a good option to enhance task completion rates and improve their utility. Hereafter, we will introduce our proposed DQNP-DRO algorithm.

Our studied problem represents as a Markov Decision Process (MDP), for which Deep Q-Network (DQN)\cite{DQN1} offers an useful technique that integrates Q-learning with deep neural networks, as benefiting from its online learning and the function approximation capabilities. Accordingly, our proposed DQNP-DRO algorithm implements on a DQN model with prioritized experience replay (PER)\cite{PER}.
\begin{figure}[t!]
	\centering
	\includegraphics[width=0.9\columnwidth]{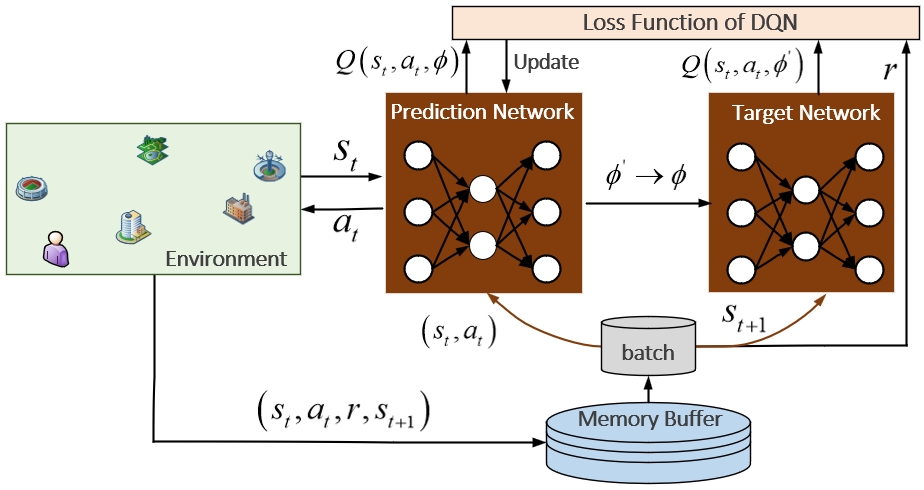}
	\caption{A detailed structure of the ST-DP$^2$WR.}
	\vspace{-0.3cm}
\end{figure}
Fig. 3 depicts the schematic of our proposed DQNP-DRO algorithm, where our considered dynamic MCS network is conceptualized as an environment with multiple agents. Meanwhile, each worker is seen as a DQN agent, making decisions on the observed state of the environment whether to continue the current task or choose next task following vector $\varphi(w_j)$. The goal of the agent is to discover an optimal policy $\pi$, maximizing long-term cumulative rewards by taking an action $\mathbbm{a}$ under a given state $\mathbbm{s}$. Optimal state-action values can be found through the Q-function based on state-action pairs, defined as
\begin{equation}\tag{34}
	Q(\mathbbm{s}_t,\mathbbm{a}_t)={\bf E}_{\pi}\left[\sum_{i=1}^{\infty}\chi^{i-1}\mathbb{R}_{t+i}|\mathbb{S}_t=\mathbbm{s},\mathbb{A}_t=\mathbbm{a}\right],
\end{equation}
where $\chi$ is a discount factor ($0\le \chi \le 1$), and $\mathbb{R}_{t}$ is a reward that the agent can obtain by taking action $\mathbbm{a}_t$ in state $\mathbbm{s}_t$. Q-learning uses a Q-table to store Q-values and the online updating rule of the Q-values with a learning rate $lr$ ($0\le lr \le 1$) is given by
\begin{equation}\tag{35}
	\begin{aligned}
		&Q(\mathbbm{s}_t,\mathbbm{a}_t)\leftarrow \\&Q(\mathbbm{s}_t,\mathbbm{a}_t)+lr\left[\mathbbm{r}_t+\chi \underset{{\mathbbm{a}}}{\max}Q(\mathbbm{s}_{t+1},\mathbbm{a}_t)-Q(\mathbbm{s}_t,\mathbbm{a}_t)\right].
	\end{aligned}
\end{equation}

The state, action and reward function for the proposed DQNP-DRO algorithm are defined as follows:

\noindent
$\bullet$ \textit{State:} the state is composed of five elements: \textit{i)} worker's current location; \textit{ii)} current timeslot; \textit{iii)} target task; \textit{iv)} information required for executing the task (e.g., the sequence of tasks, locations and time windows of tasks, and penalties for abandoning tasks); and \textit{v)} the distance from the worker to the current target task.

\noindent
$\bullet$ \textit{Action:} the action considers two main factors: \textit{i)} continue with the selected task, and \textit{ii)} abandon the current task and choose the next task following the pre-path $L_j^{\text{best}}$.

\noindent
$\bullet$ \textit{Reward:} the reward function accounts for four conditions: \textit{i)} if the worker is on its way to the task and chooses to continue, a negative reward is given for moving loss (e.g., $c_{i,j}^{\text{move}}$); \textit{ii)} if the worker plans to change the current task, a negative reward can be caused for contract break (e.g., $q_{i,j}^F$ or $q_{i,j}^S$); \textit{iii)} upon reaching the task, a negative reward is incurred for the cost of data collection and transmission (e.g., $c_{i,j}^{\text{sense}}+c_{i,j}^{\text{tran}}$); and \textit{iv)} completing the task yields a positive reward (e.g., $ p_{i,j}^F$ or $p_{i,j}^S $).

Unlike Q-learning, DQN uses a deep neural network to approximate the optimal state-action value Q-function, i.e., $Q(\mathbbm{s},\mathbbm{a},\theta^\prime)\approx \underset{\pi}{\max} ~Q(\mathbbm{s},\mathbbm{a})$, where $\theta^\prime$ denotes a set of parameters for a neural network.
The training process of the DQN includes the following steps:

\noindent
$\bullet$ Initialize the memory buffer size, structure and parameters of the DQN.

\noindent
$\bullet$ Interact with the environment and generate a set of training data. Each interaction generates a tuple $(\mathbbm{s}_t,\mathbbm{a}_t,\mathbbm{r}_t,\mathbbm{s}_{t+1})$.

\noindent
$\bullet$ Store the experience-tuples into the memory buffer.

\noindent
$\bullet$ By employing a strategy of PER\cite{PER}, a batch of experience tuples of specified size is sampled from the memory buffer. These data are then input into the DQN, where the value of the loss function can be calculated by (36)
\begin{equation}\tag{36}{\small
	\begin{aligned}
	Loss=\left[\mathbbm{r}_t+\chi\underset{\mathbbm{a}}{\max}Q(\mathbbm{s}_{t+1},\mathbbm{a}_{t+1},\phi^\prime)-Q(\mathbbm{s}_{t},\mathbbm{a}_{t},\phi^\prime)\right]^2,
	\end{aligned}}
\end{equation}
where $\phi$ and $\phi^\prime$ are parameters of a prediction network and those of a target network in the DQN, respectively. $Q(\mathbbm{s}_{t},\mathbbm{a}_{t},\phi^\prime)$ is the output of the prediction network, while $\underset{\mathbbm{a}}{\max}~Q(\mathbbm{s}_{t+1},\mathbbm{a}_{t+1},\phi^\prime)$ represents the output of the target network.

\noindent
$\bullet$ Use an RMSprop method\cite{RMS} to update the parameters of DQN until the values of the loss function converge.

\noindent
$\bullet$ Replace the parameters of the target network by that of the prediction network after a certain number of training steps.

\begin{algorithm}[t!] {\small
	\caption{{Proposed DQNP-DRO Algorithm}}
	\LinesNumbered 
	
	\textbf{Initialization:} $t=0$, $\bm{S^{\prime}}\langle t\rangle\leftarrow \varnothing$, $L_j$, initialize DQN-PER network of worker $w_j$
	
	\For{episode = 1, 2, 3, ... }{
		\For{$t = 1, 2,...,\bm{T} $}{
			
			\If{$\bm{S^{\prime}_t}\neq \varnothing$}{
				$\nu_t(w_j)\leftarrow$ ST2M matching algorithm
			}
			
			\For{$\forall w_j \in \bm{W}$}{
				$L_j \leftarrow L_j \cup \nu_t(w_j)$
				
				Observe current state $\mathbbm{s}_{t}$
				
				Select action $\mathbbm{a}_{t}$ using $\varepsilon$-greedy strategy\tcp*{Online inference}
				
				\If{$w_j$ gives up task $s_i \in \varphi(w_j)$}{
					$\bm{S^{\prime}}\langle t\rangle \leftarrow \bm{S^{\prime}}\langle t\rangle \cup \{s_i\}$
				}
				
				Get reward $\mathbbm{r}_{t}$ and next state $\mathbbm{s}_{t+1}$
				
				Store $(\mathbbm{s}_{t}, \mathbbm{a}_{t}, \mathbbm{r}_{t}, \mathbbm{s}_{t+1})$ into memory buffer
			}
			
			\If{$\forall w_j \in \bm{W}$ complete all accepted tasks}{
				\textbf{Break}
			}
		}
		
		Update the DQN-PER network of each worker $w_j$ \tcp*{Offline training}
	}
	Continue until rewards converge
}
\end{algorithm}

{We next detail our proposed DQNP-DRO algorithm (given by Alg.~3):

\noindent
\textbf{Step 1. Initialization:} We first initialize the relevant parameters, where timeslot $t=0$, the set of tasks requiring temporary recruitment $\bm{S^{\prime}}\langle t\rangle$ is an empty set, and the DQN-PER network is initialized for each worker $w_j\in\bm{W}$.

\noindent
\textbf{Step 2. Temporary task-worker determination:} When there exist tasks with unsatisfying service quality (line 4, Alg.~3), we employ the ST-M2M matching algorithm to allocate temporary workers, yielding the task set \(\nu_t(w_j)\). These tasks are then appended to worker $w_j$’s pre-planned trajectory $L_j$ (line 7, Alg.~3). As ST-M2M is structurally similar to FT-M2M, we defer its details to Appx.~F.

\noindent
\textbf{Step 3. Online inference:} At each timeslot, each worker \(w_j\) observes its current state $\mathbbm{s}_t$, selects an action $\mathbbm{a}_t$ based on the $\varepsilon$-greedy strategy, and receives the reward $\mathbbm{r}_t$ as well as the updated state $\mathbbm{s}_{t+1}$ (lines 8–9, Alg.~3). The worker then stores the tuple $(\mathbbm{s}_t, \mathbbm{a}_t, \mathbbm{r}_t, \mathbbm{s}_{t+1})$ into its local memory buffer. This online inference process is repeated at each timeslot for task execution and decision-making.

\noindent
\textbf{Step 4. Offline training:} After completing one episode, each worker updates its DQN-PER network based on the stored experiences using PER. The training aims to minimize the loss function in (36), and continues until convergence of the reward function. Once training is complete, the parameters of the target network are synchronized with those of the prediction network.}

\section{Evaluation}
We conduct comprehensive evaluations to verify the effectiveness of our methodology, which are carried out via Python 3.9 with 13th Gen Intel Core i9-13900K*32 and NVIDIA GeForce RTX 4080. Recall the previous sections, our proposed FT-SMP$^3$ and ST-DP$^2$WR work together to facilitate a stagewise trading paradigm, which are collectively abbreviated as ``StagewiseTM$^3$atch'' in this section, for analytical simplicity.
\subsection{Simulation Settings}
{To ensure both realism and analytical flexibility, we adopt a hybrid evaluation methodology that combines real-world mobility data with controllable system variables.
Specifically, we utilize the real-world dataset that records taxi rides in Chicago from 2013 to 2016 across 77 community areas~\cite{dataset}. This dataset has been widely adopted in previous MCS studies~\cite{RWork_stable 5,taxidata, DP2} as a proxy for worker mobility, and thus offers realistic urban dynamics. Following~\cite{DP2}, we select the 77-th community area as our sensing region and use a total of 271259 mobility records\footnote{This dataset enables the construction of dynamic worker distributions and mobility patterns, providing realism in spatial-temporal task-worker interactions.}.}


{To approximate the behavior of workers, we extract features such as the pick-up locations of taxis as the initial worker positions (e.g., $l^{\text{w}}_j\langle 0\rangle$), and use the average taxi travel speed to estimate worker movement velocity (e.g., $v_j$). Furthermore, we map taxi fare data to the travel-related energy cost of workers (e.g., $e^{\text{m}}_j$). These mappings allow us to estimate worker cost parameters based on empirical data, rather than arbitrary values. Meanwhile, MCS tasks are randomly distributed within the designated sensing region to emulate spatial heterogeneity. Accordingly, key parameters are shown by Table 2. }
\begin{table}[htb]\vspace{-0.3cm}
	{{\scriptsize
		\caption{\footnotesize{Key simulation setting}}\vspace{-0.6cm} 
			\begin{center}
				\begin{tabular}{|c|c|c|c|}
					\hline
					\multicolumn{1}{|c|}{\textbf{ Parameter}} & \multicolumn{1}{c|}{\textbf{Value}}&\multicolumn{1}{|c|}{\textbf{ Parameter}} & \multicolumn{1}{c|}{\textbf{Value}} \\ \hline
$\bm{T}$ & 100& $\mu_1, \mu_2$ & \makecell[c]{150, 400\cite{Future Trading 2,QHY2}\\ (SNR $\approx$ [17, 23] dB)} \\
\hline
$\alpha_{i,j}$ & [0.005, 0.01]&$e_j^{\text{c}}$ & [0.002, 0.006] \$/timeslot\cite{canshu} \\
\hline
$\bm{K}$ & 20&$\tau_{i,j}^{\text{D}}$ & [1, 5] timeslots \\
\hline
$t_i^{\text{b}}$ & [1, 90]& $t_i^{\text{e}}$ & $[10+t_i^{\text{b}}, 99+t_i^{\text{b}}]$, s.t. $t_i^{\text{e}} \leq \bm{T}$\\
\hline
$B_i$ & [30, 45]\cite{DP2}&$f_j$ & [256, 512] Mb/timeslot\\
\hline
$d_i$ & [2, 4] Gb&$e_j^{\text{t}}$ & [450, 550] mW\cite{QHY2} \\

\hline
$\rho_1,...,\rho_5$ & 0.3 \cite{QHY2} &$p_{i,j}^{\text{Desire}}$&[8,10]\\
\hline
$Q_i^{\text{D}}$ & [10, 15]\cite{DP2}&$v_j$ & [100, 200] m/timeslot \\
\hline
$W$ & 6 MHz \cite{Future Trading 2}&$e_j^{\text{m}}$ & [0.02, 0.05] \$/timeslot\cite{canshu} \\
\hline
$\triangle p$ &0.5 &$\mathbb{V}_3$ & [0,1]  \\
\hline
				\end{tabular}
			\end{center}
	}\vspace{-0.1cm}}
\end{table}

Moreover, for definiteness and without loss of generality, we are inspired by the Monte Carlo method and conduct 100 simulations for each figure in this section.
\subsection{Benchmark Methods and Evaluation Metrics}
To conduct better evaluations, we involve comparable benchmark methods from diverse perspectives. Regarding \textit{the conventional spot trading mode}, we consider spot trading-enabled M2M matching (abbreviated as ConSpotTM$^3$atch), which borrows the idea from \cite{RWork_dynamic 1}, mapping each worker to a vector of tasks by analyzing the current network conditions. To underscore the \textit{importance of meeting diverse demands of spatial-temporal tasks}, we consider:
\textit{i)} Spot trading-enabled M2M matching without considering spatial-temporal demands of task (abbreviated as ConSpotTM$^3$atch\textbackslash TD), as inspired by \cite{RWork_stable 5}. This variant mirrors ConSpotTM$^3$atch but disregards the spatial and temporal constraints of tasks. \textit{ii)} The method similar to our StagewiseTM$^3$atch without considering spatial-temporal demands of tasks, in which the long-term contracts do not take into account penalties (abbreviated as StagewiseTM$^3$atch\textbackslash TD), as inspired by \cite{DP2}.
 To emphasize \textit{the importance of risk control in dynamic networks}, we introduce a benchmark called stagewise trading-enabled M2M matching without risk analysis (abbreviated as StagewiseTM$^3$atch\textbackslash Risk), which is analogous to our StagewiseTM$^3$atch but excludes risk analysis (i.e., constraints (22c), (23c), (23d), (33c), and (33d)). To explore \textit{the trade-off between time efficiency and resource allocation performance}, we introduce another two benchmarks under spot trading mode: service quality-preferred method (abbreviated as Quality\_P) and random matching (abbreviated as Random\_M), borrowing the ideas from \cite{baseline_random}, where tasks in Quality\_P prefer workers with the highest service quality under budget constraints; while Random\_M randomly selects workers under budget constraints.

To conduct quantitative evaluations, we also focus on crucial performance metrics detailed in the following: 

\noindent $\bullet$ \textbf{Service quality:} As one of the most important factors in MCS networks, service quality is calculated by the overall received service quality of tasks.

{\noindent $\bullet$ \textbf{Utility of tasks and workers:} The overall utilities received by tasks and workers.}

\noindent $\bullet$ \textbf{Social welfare:} The summation of utilities of both tasks and workers.

\noindent $\bullet$ \textbf{The proportion of tasks that meet their desired service quality (PoDSQ):} PoDSQ represents the ratio of tasks that meet their desired service quality to the overall tasks.

\noindent $\bullet$ \textbf{The proportion of tasks abandoned by workers (PoTAW):} PoTAW represents the ratio of tasks abandoned by workers to the overall tasks accepted by workers.

\noindent $\bullet$ \textbf{Running time (RT, ms):} The running time is obtained by Python on verison 3.9, reflecting time efficiency.

\noindent $\bullet$ \textbf{Number of interactions (NI):} Total number of interactions between tasks (owners) and workers to obtain matching decisions, reflecting the overhead on decision-making.

{\noindent $\bullet$ \textbf{Delay Incurred by Participant interactions (DIP):} Total delay caused by price negotiation and service coordination, simulated using real-world wireless latency distributions (uplink: [0.5, 11] ms; downlink: [0.5, 4] ms) as suggested by \cite{delay1} and \cite{delay2}.

\noindent $\bullet$ \textbf{Energy Consumption Incurred by Participant interactions (ECIP):} Total energy cost during interaction processes, calculated based on power consumption models of mobile worker devices ([0.2, 0.4] W) and task-owner devices ([6, 20] W), following the measurement framework in \cite{DP2} and \cite{power}.}

\subsection{Performance Evaluations}
\subsubsection{Service quality, utility, and social welfare}
\begin{figure*}[]{
	\vspace{-0.2cm}
	\centering
	\setlength{\abovecaptionskip}{-1 mm}
	\includegraphics[width=2\columnwidth]{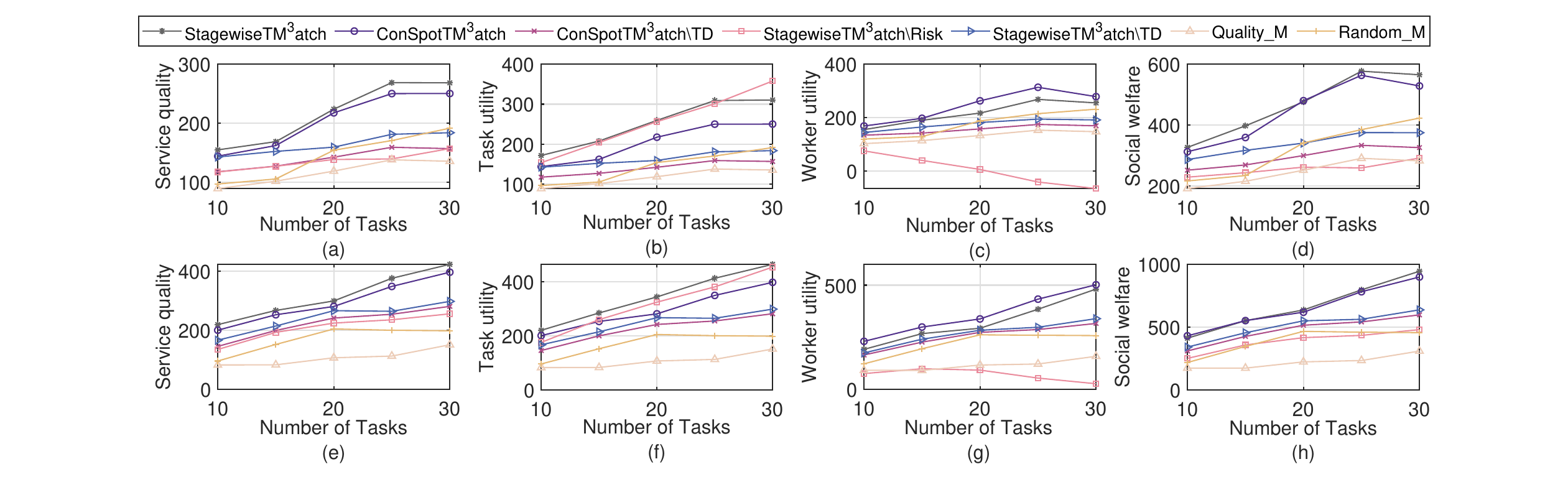}
	\caption{Performance comparisons in terms of overall utility, overall service quality and social welfare under different problem sizes, where (a)-(d) consider 15 workers, and (e)-(h) consider 30 workers in the network.}
	\vspace{-0.5cm}}
\end{figure*}
%
{We study the performance of our proposed StagewiseTM$^3$atch method in terms of overall utility of workers and tasks, service quality and social welfare, in Fig. 4.} To capture various problem scales, Figs. 4(a)-4(d) consider 15 workers, while Figs. 4(e)-4(h) consider 30 workers.

Fig. 4(a) illustrates that our proposed StagewiseTM$^3$atch method achieves the highest service quality, thanks to the stagewise mode that comprehensively considers task diversity and risks in dynamic networks. In contrast, ConSpotTM$^3$atch performs slightly worse, because it only employs spot trading mode. Furthermore, our StagewiseTM$^3$atch significantly outperforms the other four benchmark methods, due to the drawbacks in their fundamental matching logic. For instance, ConSpotTM$^3$atch\textbackslash TD and StagewiseTM$^3$atch\textbackslash TD neglect the diversity of task requirements, StagewiseTM$^3$atch\textbackslash Risk overlooks risk management, and Quality\_P focuses solely on task utility, thus significantly increasing the probability of failures in performing tasks. Additionally, due to the inherent randomness of Random\_M, it faces unsatisfactory service quality of tasks. Note that in Fig. 4(a), the curves of StagewiseTM$^3$atch and ConSpotTM$^3$atch slightly decline after 25 tasks, since the time window limits the number of tasks that workers can complete. The performance shown in Fig. 4(e) is similar to that of Fig. 4(a). However, after 25 tasks, Fig. 4(e) features more workers and can serve more tasks compared to Fig. 4(a). Therefore, as the number of tasks increases, the curves for StagewiseTM$^3$atch and ConSpotTM$^3$atch in Fig. 4(e) continue to rise.

We next show Fig. 4(b) and Fig. 4(f) to evaluate the performance on tasks' utility, consisting of received service quality and the remaining financial budget (further reflects the costs for purchasing services). Methods such as ConSpotTM$^3$atch, ConSpotTM$^3$atch\textbackslash TD, StagewiseTM$^3$atch\textbackslash TD, Quality\_P, and Random\_M do not incorporate compensations for tasks, thus enabling performance similar to that in Fig. 4(a) and Fig. 4(e). In contrast, StagewiseTM$^3$atch compensates those tasks with default workers, thereby allowing more budget to recruit temporary workers. Notably, StagewiseTM$^3$atch\textbackslash Risk neglects risk analysis for both workers and tasks, resulting in a high rate of trading failures. Consequently, workers are required to compensate task owners for contract breaches. This, in turn, leads to a scenario where StagewiseTM$^3$atch\textbackslash Risk provides relatively lower service quality for tasks while yielding relatively high task's utility.
As for workers, Fig. 4(c) and Fig. 4(g) reveal that ConSpotTM$^3$atch achieves the best performance on their utilities, which attributed to the consideration of task diversity and risk management, striving to complete more matched tasks successfully. Our StagewiseTM$^3$atch achieves slightly lower workers' utility than that of ConSpotTM$^3$atch, since workers in our consideration may have to make compensations to tasks. Besides, StagewiseTM$^3$atch\textbackslash Risk method accepts a high number of tasks without analyzing risks. Thus, more service failures can be incurred with the rising number of tasks, e.g., workers have to pay to those tasks they can not complete, due to factors such as delay events. Other benchmark methods, either failing to adequately consider task diversity or employing simplistic matching strategies, result in a significant reduction in the number of tasks completed by workers, thus affecting the utility of the workers.

In Fig. 4(d) and Fig. 4(h), both StagewiseTM$^3$atch and ConSpotTM$^3$atch achieve the best performance in terms of social welfare and exhibit comparable results.  
Notably, ConSpotTM$^3$atch requires making trading decisions for every practical transaction, leading to excessive delay on decision-making. As a result, it exhibits inferior performance compared to StagewiseTM$^3$atch in NI and RT, as illustrated in Fig. 6, with further details discussed in Sec. 6.3.3.  
Furthermore, thanks to the ST-M2M mechanism of StagewiseTM$^3$atch, when an MCS task is abandoned, the task provider can still initiate temporary recruitment to enhance PoDSQ (see Fig. 5(a)). Additionally, our robust risk management mechanism ensures superior PoTaW performance for StagewiseTM$^3$atch (see Fig. 5(b)). These advantages collectively lead to StagewiseTM$^3$atch outperforming ConSpotTM$^3$atch in both PoDSQ and PoTaW. {Besides, we also investigate the performance on average task's utility and service quality per task under varying numbers of tasks, as shown in Fig.~8 of Appx. G.2, to reveal our efficiency under limited resources, due to space limication.}


\subsubsection{PoDSQ and PoTAW}
\begin{figure}[]{
	\vspace{-0.1cm}
	\centering
	\setlength{\abovecaptionskip}{-1 mm}
	\includegraphics[width=1\columnwidth]{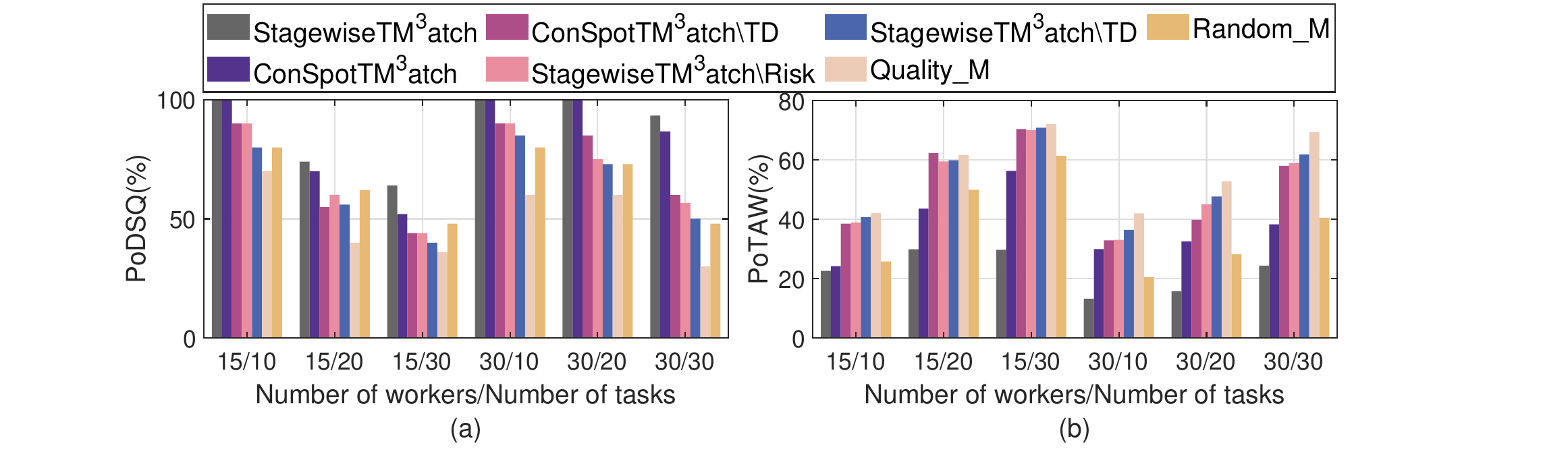}
	\caption{Performance comparisons in terms of PoDSQ and PoTAW.}
	\vspace{-0.5cm}}
\end{figure}
As the uncertain factors can prevent workers from completing assigned tasks, further leading to trading failures and unsatisfying service quality, we involve PoDSQ and PoTAW as two crucial significant factors, upon having various problem scales (see Fig. 5).

As shown in Fig. 5(a), our method StagewiseTM$^3$atch outperforms other methods on FoSDQ, as attributed to the well-designed stagewise mechanism, as well as the effective risk analysis, allowing our proposed method to surpass the service quality of the existing advanced ConSpotTM$^3$atch. In contrast, since ConSpotTM$^3$atch\textbackslash TD and StagewiseTM$^3$atch\textbackslash Risk overlook diverse task demands and risk controls, the values of FoSDQ of them stay lower. Furthermore, the inherent factor on greediness and randomness of Quality\_P and Random\_M bring them with less satisfactory outcomes, further reflecting poor trading experience of task owners.

	Since this paper considers real-world factors in MCS networks, e.g., an uncertain delay event and its duration, workers may be unable to catch the deadline of every task, thus forcing them to give up some assigned tasks. To evaluate the task abandon rate, which reflects the trading experience for both workers and task owners in the considered MCS networks, we conduct Fig. 5(b). In this figure, we can clearly see that our proposed StagewiseTM$^3$atch consistently outperforms other methods on PoTAW across various problem scales, maintaining a rate below 30\%. This impressive performance benefits by our consideration on risk management, to ensure that all matched tasks are controllable and achievable with an acceptable probability. ConSpotTM$^3$atch performs a lower PoTAW than ours due its unilateral concern on a single-stage mechanism, lacking the backup for tasks to recruit more workers, when some workers are facing difficulties to catch the assigned tasks on its half-way, and thus decide to give up. Moreover, StagewiseTM$^3$atch\textbackslash Risk lacks designs on risk control, ConSpotTM$^3$atch\textbackslash TD and StagewiseTM$^3$atch\textbackslash TD overlook diverse demands of tasks, while Quality\_P and Random\_M rely solely on greedy and random strategies for task assignment. These limitations can cause a large number of unachievable tasks accepted by workers, thereby leading to poor performance on PoTAW.
\subsubsection{RT, NI, DIP, and ECIP}
\begin{figure} {\centering 
	\subfigtopskip=2pt
	\subfigbottomskip=10pt
	\setlength{\abovecaptionskip}{-0.1cm}
	\subfigure[] {
		\label{fig:a} 
		\includegraphics[width=0.53\columnwidth]{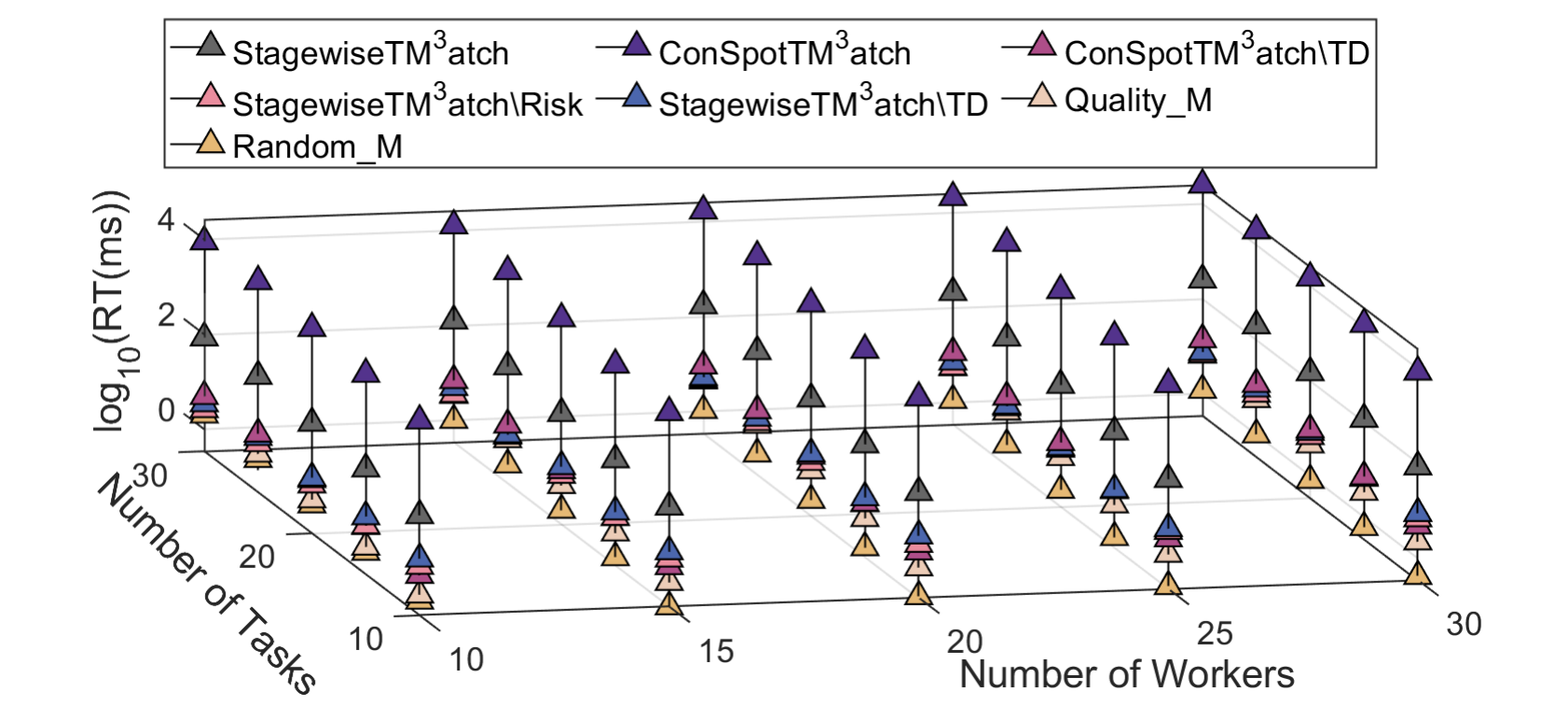} 
	} \hspace{-5.6mm}
	\subfigure[] { 
		\label{fig:b} 
		\includegraphics[width=0.46\columnwidth]{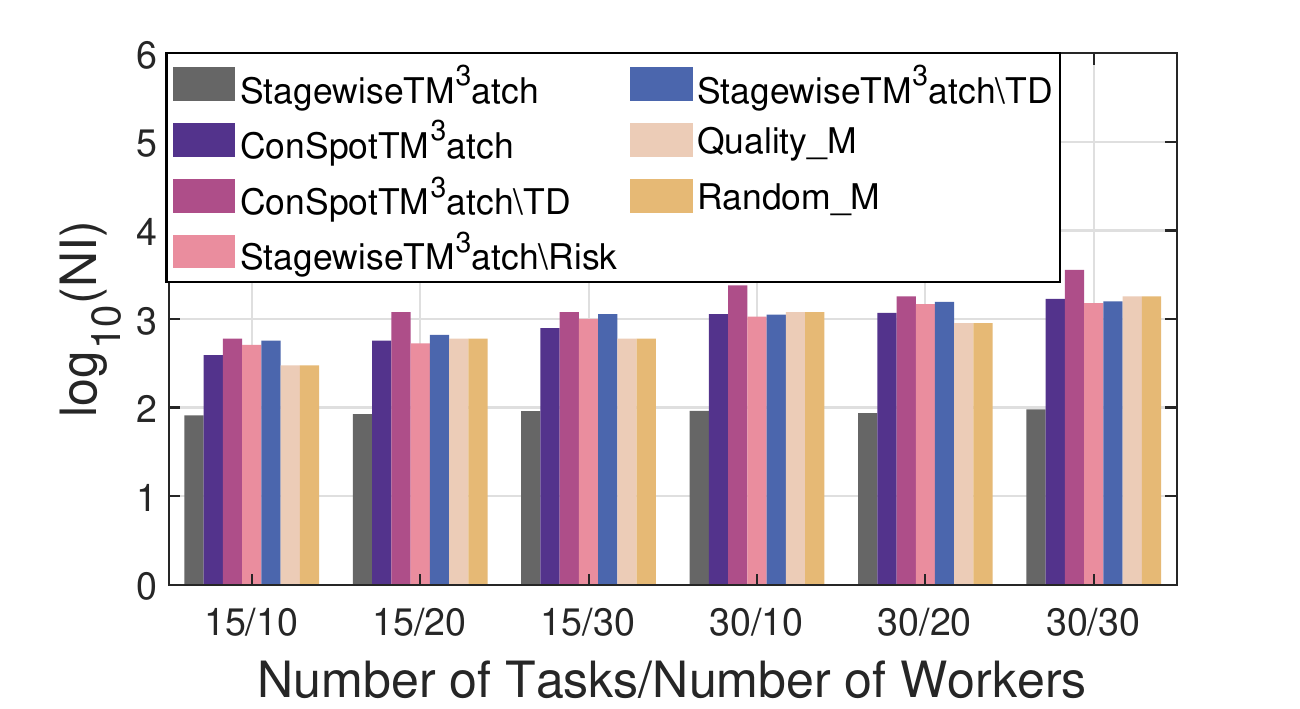} 
	}  
    	\hspace{-5.8mm}\subfigure[] {
		\label{fig:a} 
		\includegraphics[width=0.46\columnwidth]{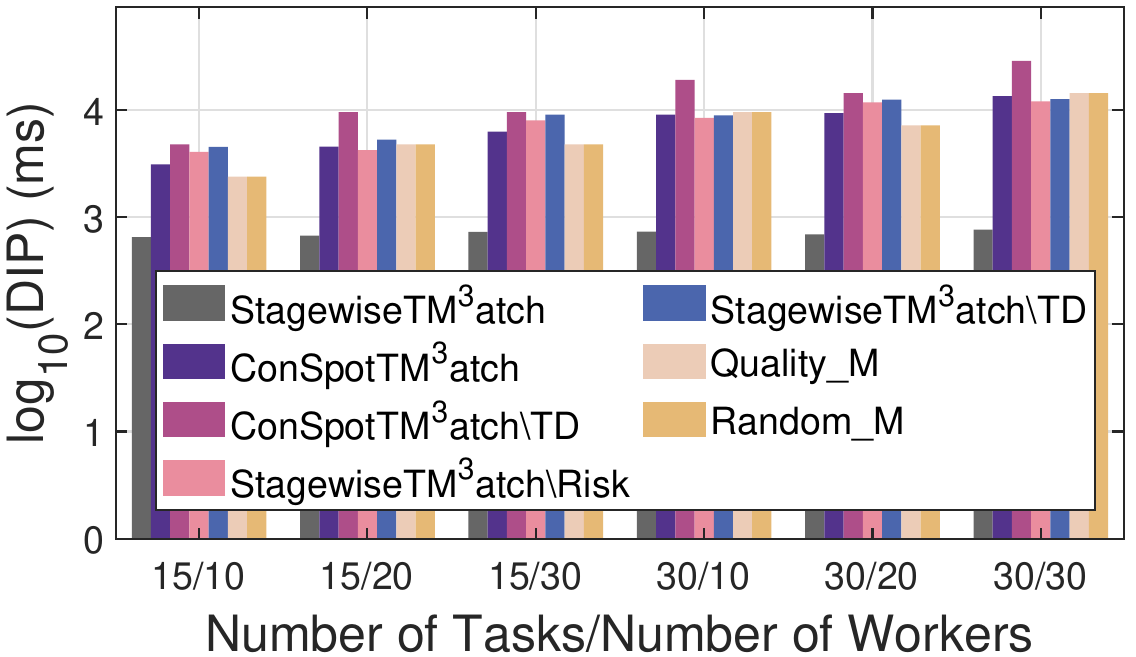} 
	} \hspace{-2mm}
	\subfigure[] { 
		\label{fig:b} 
		\includegraphics[width=0.46\columnwidth]{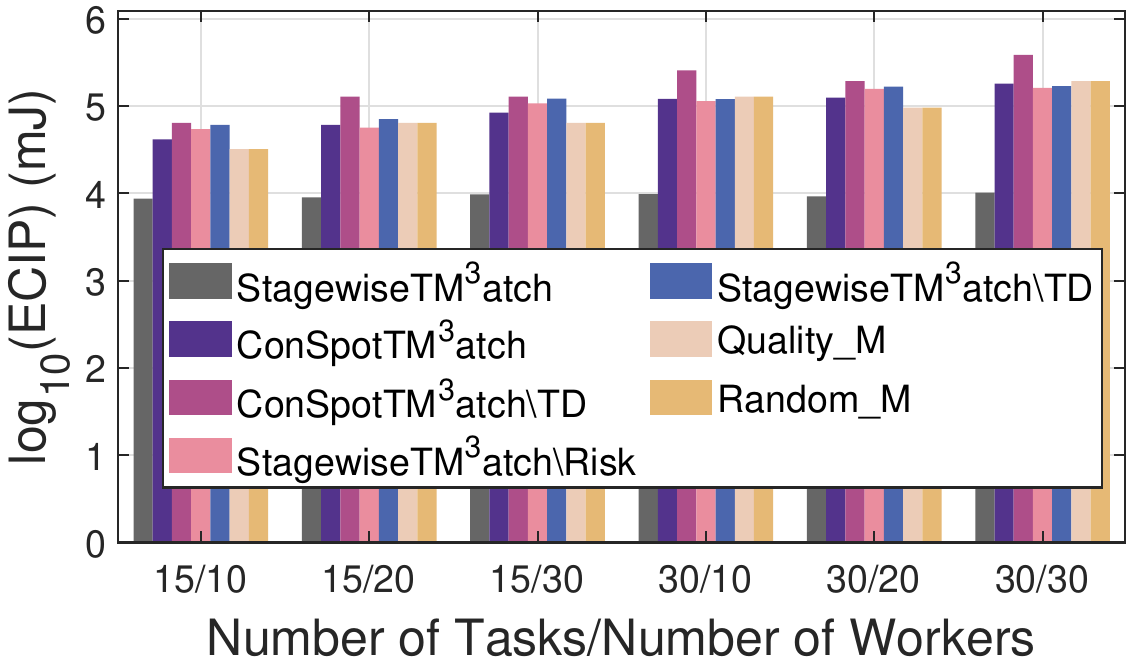} 
	} 
	\caption{Performance comparisons in terms of RT, NI, DIP, and ECIP.} 
	\label{fig} 
	\vspace{-0.5 cm}}
\end{figure}
{Efficiency—particularly in terms of time and energy—serves as a primary performance indicator in dynamic service trading markets. To evaluate this aspect, we examine Four key metrics: RT (e.g., the time consumed for decision-making), NI (e.g., the number of negotiation rounds regarding factors such as task/worker selection and trading price), {DIP (e.g., the delay incurred by participant interactions such as price negotiation and service coordination), and ECIP (e.g., the energy consumption during these decision-making procedures)}, as illustrated in Fig. 6. For enhanced clarity, a logarithmic scale is adopted in the figure to more clearly highlight performance gaps among different methods.}

As observed in Fig. 6(a), ConSpotTM$^3$atch always stays large because it matches tasks to workers during each transaction, based on the current network/market conditions, leading to excessive delay on decision-making. While the RT of our StagewiseTM$^3$atch remains stably lower than ConSpotTM$^3$atch, even having a rising number of tasks and workers. This is because each worker has pre-matched tasks and recommended paths thanks to the pre-signed long-term contracts in futures trading stage, whereas in the spot trading stage, facing the uncertainties of a dynamic networks, workers only need to make minor adjustments. Then, although ConSpotTM$^3$atch\textbackslash TD, StagewiseTM$^3$atch\textbackslash TD and StagewiseTM$^3$atch\textbackslash Risk outperform ConSpotTM$^3$atch and StagewiseTM$^3$atch in RT, they overlook crucial factors such as spatial-temporal demands of tasks and risks, thus undergoing unsatisfying performance on utility of workers and tasks, service quality, and social welfare. Furthermore, due to simple ideas adopted by Quality\_P and Random\_M, these methods exhibit far lower RT as compared to others. However, they exhibit poor performance on service quality (see Figs. 4-5). 

{Then, Fig. 6(b)–(d) present the performance for NI, DIP, and ECIP, capturing the decision-making overhead such as the energy and delay consumed during trading interactions between workers and task owners, upon having different numbers of tasks and workers. As shown in Fig. 6(b), our proposed StagewiseTM$^3$atch significantly outperforms all benchmark methods. This advantage primarily stems from the fact that, in ConSpotTM$^3$atch and ConSpotTM$^3$atch\textbackslash TD, participants are required to negotiate trading decisions (e.g., task allocation and service pricing) during each individual transaction, inevitably incurring considerable time and energy overhead. In particular, ConSpotTM$^3$atch\textbackslash TD and StagewiseTM$^3$atch\textbackslash TD exhibit the highest values of NI, as they fail to account for the heterogeneous demands of tasks, thereby intensifying competition among workers.
In contrast, StagewiseTM$^3$atch promotes participation in futures trading, which substantially improves decision-making efficiency. Furthermore, it surpasses StagewiseTM$^3$atch\textbackslash Risk, benefiting from its carefully designed risk management strategies.
Although Quality\_P and Random\_M avoid complex bargaining and risk management by enabling only basic interactions (i.e., workers express their willingness to participate, and tasks select preferred workers), such interactions still incur overhead—especially as the market scales up\footnote{It is important to note that for NI, our StagewiseTM$^3$atch outperforms Quality\_P and Random\_M because we count only the number of interaction rounds required for decision-making, not the time spent in each round. As StagewiseTM$^3$atch performs pre-determined planning for tasks and paths, it achieves lower NI despite incurring slightly higher RT c
ompared to Quality\_P and Random\_M.}.
Overall, StagewiseTM$^3$atch consistently demonstrates superior performance in terms of NI, even when compared to simplified benchmark methods.} {In Figs.~6(c)–6(d), DIP and ECIP exhibit trends consistent with NI, as higher numbers of negotiation rounds directly lead to greater delay and energy consumption. By minimizing NI through efficient futures trading and well designed risk management strategies, StagewiseTM$^3$atch maintains the best overall performance in both metrics.}

{Regarding simulations of the individual rationality of tasks and workers, as well as the impact of the asked payment decrement $\triangle p$ on the trade-off between convergence speed and matching outcomes, they have been moved to Appx.~G due to space limitation.}

{\subsection{Larger-Scale Problem Testing}
\begin{table}[htb]
	{\scriptsize 
		\caption{Performance Evaluations for Larger-Scale Problems (Alg 1: StagewiseTM$^3$atch, Alg 2: ConSpotTM$^3$atch, Alg 3: ConSpotTM$^3$atch\textbackslash TD, Alg 4: StagewiseTM$^3$atch\textbackslash Risk, Alg 5: StagewiseTM$^3$atch\textbackslash TD, Alg 6: Quality\_M, Alg 7: Random\_M)}
		\begin{center}
			\setlength{\tabcolsep}{0.5mm}{\vspace{-0.2cm}
				\begin{tabular}{|cccccccc|}
					\hline
					\textbf{Performance} & \textbf{Alg 1}& \textbf{Alg 2}&\textbf{Alg 3}&\textbf{Alg 4} & \textbf{Alg 5} & \textbf{Alg 6}& \textbf{Alg 7}\\ \hline
					Servise quality&1279.96 & 1202.66 & 510.42 & 416.58 & 539.91 & 749.17 & 494.56
					\\ \hline Task's utility &1568.73& 1202.66& 510.42& 1855.91& 539.91& 749.17& 494.56
					\\ \hline
					Worker's utility &1243.85& 1495.13& 456.58 & -1055.85& 526.58& 286.09& 903.66
					\\ \hline
					Social welfare &2812.58&2697.79&967.00 &800.06&1066.49&1035.26&1398.22\\ \hline
					NI &237.24&38204.51&48230.23&26439.08&25479.26&24000&24000
					\\ \hline
					RT (ms)  &8885.13&866600.26&1570.58&1402.59&1156.46&168.65&20.66	\\ \hline
			\end{tabular}}
	\end{center}}\vspace{-.5cm}
\end{table}
To facilitate comprehensive evaluations, we conduct large-scale assessments by considering an expanded set of 100 tasks and 150 workers. As shown in Table 3, our proposed StagewiseTM$^3$atch demonstrates superior performance in terms of service quality for MCS tasks and social welfare. Moreover, it shows commendable performance on tasks' utility, workers' utility, as well as time and energy overhead (i.e., RT and NI) associated with decision-making, making it a flexible and applicable reference for future dynamic and large-scale networks.}

\section{Conclusion}
We investigate a novel stagewise trading framework that integrates futures and spot trading to facilitate efficient and stable matching between diverse tasks and workers, in a dynamic and uncertain MCS network. In the former stage, we propose FT-SMP$^3$ for long-term task-worker assignment and path pre-planning for workers based on historical statistics and risk analysis. The following stage investigates ST-DP$^2$WR mechanism to enhance workers' and tasks' practical utilities by facilitating temporary worker recruitment. Theoretical exploration demonstrates that our proposed mechanisms can support essential properties including individual rationality, strong stability, competitive equilibrium, and weak Pareto optimality. Evaluations through real-world dataset demonstrate our superior performance in comparison to existing methods across various metrics. We are also interested in exploring smart contract design for service trading and potential collaborations among workers, as interesting future directions.

\newpage
\clearpage
\appendices
\setcounter{equation}{36}
\section{Key Notations}
Key notations in this paper are summarized in Table 4.

\begin{table*}[b!]
	{\footnotesize
		\caption{\footnotesize{Key notations}}\vspace{-0.3cm} 
			\begin{center}
				\begin{tabular}{|l|l|}
					\hline
					\multicolumn{1}{|l|}{\textbf{Notation}} & \multicolumn{1}{l|}{\textbf{Explanation}} \\ \hline
					$ \bm{S} $, $ \bm{W} $, $ \bm{S^{\prime}}\langle t\rangle$, $\bm{W^{\prime}}\langle t\rangle$ & The MCS sensing task set and worker set in FT-SMP$^3$, ST-DP$^2$WR \\ \hline
					$ s_i $, $ w_j $ & The $ i^\text{th} $ sensing task in $ \bm{S} $, $ \bm{S^{\prime}}\langle t\rangle$, the $ j^\text{th} $ worker in $ \bm{W} $, $\bm{W^{\prime}}\langle t\rangle$ \\\hline
					$ p_{i,j}^F $, $ p_{i,j}^S $ & Payment of task $ s_i $ offered to worker $ w_j $ in FT-SMP$^3$, ST-DP$^2$WR\\\hline
					$ q_{i,j}^F $, $ q_{i,j}^S $ & Compensation worker $w_j$ offered to task $ s_i $ in FT-SMP$^3$, ST-DP$^2$WR\\\hline
					$ t_i^{\text{b}}$, $t_i^{\text{e}} $ & the start and closing time of $s_i$\\\hline
					$ Q^{\text{D}}_i $ &Desired utility for task $ s_i $\\\hline
					$ B_{i} $& Budget for task $ s_i $\\\hline
					$l_i^{\text{s}}$, $l_j^{\text{w}}$& Location of task $s_i$ and worker $w_j$ \\\hline
					$d_i $& Data size that each worker needs to provide task $s_i$ \\\hline
					$e_j^{\text{c}}$, $e_j^{\text{D}}$, $e_j^{\text{t}}$, $e_j^{\text{m}}$& \makecell[l]{Cost consumed in each timeslot for data collection, delay event, data transmission, traveling to the target task}\\\hline
					$f_j$& \makecell[l]{The size of data collected by worker $w_j$ in each timeslot} \\\hline
					$v_j$& Movement speed of worker $ w_j $ \\\hline
					$\alpha_{i,j}$& \makecell[l]{Random variable describes encounter delay event of worker $ w_j $ during moving to task $s_i$ } \\\hline
					$\tau^{\text{delay}}$& Random variable describes uncertain duration of the delay event: \\\hline
					$\gamma_{i,j}$& \makecell[l]{Random variable describes time-varying channel qualities between worker $w_j$ and task $s_i$} \\\hline
					$ c_{i,j}\langle t^{\text{ini}} \rangle $	&\makecell[l]{Service cost of worker $ w_j $ complete task $ s_i $} \\\hline
					$ \tau_{i,j}\langle t^{\text{ini}} \rangle $&\makecell[l]{The number of timeslots required by worker \( w_j \) to complete task \( s_i \).} \\\hline
					 \( \beta_{i,j} \), \( \beta^S_{i,j} \) &\makecell[l]{Indicator of whether $w_j$ will break the contract with $s_i$ in FT-SMP$^3$, ST-DP$^2$WR} \\\hline
					\makecell[l]{$\varphi\left(s_i\right)$, $\varphi\left(w_j\right)$} & \makecell[l]{The set of workers recruited for processing task $ s_i $ and the set of tasks assigned to worker $ w_j $ in FT-SMP$^3$} \\\hline
					\makecell[l]{$\nu_t\left(s_i\right)$, $\nu_t\left(w_j\right)$} & \makecell[l]{The set of workers recruited for processing task $ s_i $ and the set of tasks assigned to worker $ w_j $ in ST-DP$^2$WR} \\\hline
					$ {U^{S}} $, $ \overline{U^{S}} $& Utility and expected utility of a task \\\hline
					$ {U^{W}} $, $ \overline{U^{W}} $& Utility and expected utility of a worker  \\\hline
					$ {R^{W}_1}$, $R_2^W$, $ {R^{S}} $& Risk associated with workers, tasks in FT-SMP$^3$
					\\
					\hline
					$ {R^{W\prime}_1}$, $R_2^{W\prime}$ & Risk associated with workers in ST-DP$^2$WR
					
					\\
					\hline
				\end{tabular}
			\end{center}
	}\vspace{-0.6cm}
\end{table*}

\section{Derivations Associated with FT-SMP$^3$}
\subsection{Derivations related to workers}
\noindent\textbf{Mathematical expectation of $c_{i,j} \langle t^{\text{ini}} \rangle$.} Random variable $\alpha_{i,j}$ follows a Bernoulli distribution $\alpha_{i,j}\sim {\bf B}\left\{(1,0),(a_{i,j},1-a_{i,j})\right\}$, with its expectation calculated by $\text{E}[\alpha_{i,j}]=1\times a_{i,j} + 0\times (1-a_{i,j})=a_{i,j} $. Moreover, $\tau^{\text{delay}}$ obeys an uniform distribution, e.g., \(\tau^{\text{delay}} \sim \mathbf{U}(t^{\text{min}}, t^{\text{max}})\), while $ \text{E}[\tau^{\text{delay}}]$ can accordingly be expressed as $\frac{t^{\text{min}}+t^{\text{max}}}{2}$. Similarly, the expected value of \(\gamma_{i,j}\) can be expressed as $\frac{\mu_1+\mu_2}{2}$. Thus, we have the value of $c_{i,j}\langle t^{\text{ini}} \rangle$ as
\begin{equation}\label{key}
	\begin{aligned}
		\text{E}[c_{i,j}\langle t^{\text{ini}} \rangle]=&c_{i,j}^{\text{move}}\langle t^{\text{ini}} \rangle+c_{i,j}^{\text{D}}+c_{i,j}^{\text{sense}}+c_{i,j}^{\text{tran}}\\=&c_{i,j}^{\text{move}}\langle t^{\text{ini}}\rangle+c_{i,j}^{\text{sense}}+\text{E}[c_{i,j}^{\text{D}}]+\text{E}[c_{i,j}^{\text{tran}}]
		\\=&c_{i,j}^{\text{move}}\langle t^{\text{ini}}\rangle+c_{i,j}^{\text{sense}}+\frac{\mathbb{V}_1e_j^{\text{t}} d_i}{W\text{log}_2 \left(1+e_j^{\text{t}} \text{E}[\gamma_{i,j}]\right)}+\\&\sum_{n=1}^{\tau_{i,j}^{\text{move}}\langle t^{\text{ini}} \rangle}\text{E}[\alpha_{i,j}]\text{E}[\tau^{\text{delay}}_n]
		\\=& c_{i,j}^{\text{move}}\langle t^{\text{ini}}\rangle+c_{i,j}^{\text{sense}}+\frac{\mathbb{V}_1e_j^{\text{t}} d_i}{W\text{log}_2 \left(1+e_j^{\text{t}} \frac{\mu_1+\mu_2}{2}\right)}+\\&\sum_{n=1}^{\tau_{i,j}^{\text{move}}\langle t^{\text{ini}} \rangle}a_i\frac{t^{\text{min}}+t^{\text{max}}}{2}
	\end{aligned}
\end{equation}

\noindent\textbf{Mathematical expectation of $c^{\text{part}}_{i,j} \langle t^{\text{ini}} \rangle$.}
During the process when worker \(w_j\) is performing task \(s_i\), each timeslot can describe different scenarios due to delay events. Analyzing each scenario results in significant computational burdens, particularly when dealing with a large number of timeslots, rendering effective computation impractical. To facilitate our analysis and ensure the individual rationality of workers, we approximate \(E[c^{\text{part}}_{i,j}]\) to \(E[c_{i,j}]\) as \(E[c^{\text{part}}_{i,j}] < E[c_{i,j}]\). This approximation ensures that workers can better manage the risk on obtaining undesirable utility when determining their asked payment.

\noindent\textbf{Mathematical expectation of $\beta_{i,j}$.}
	We use \( \bm{\mathcal{C}} \) to represent the set of possible \textbf{t}ask \textbf{c}ompletion \textbf{s}cenario (TCS). Specifically, different TCSs involve different ways that a task $t_i$ can be completed. For example, a worker moves to $t_i$ from its current location need two hops, each taking one timeslot. Correspondingly, the data collection also requires one timeslot. The current timeslot in this example is assumed to be four timeslots away from the closing time of task \( t_i \). Therefore, there are three possible TCS scenarios: 

\textit{i)} The worker encounters no delay event: The worker's movement takes two timeslots and its data collection process takes one timeslot; 

\textit{ii)} The worker encounters a delay event during the first hop, which costs one extra timeslot. During the second hop, it encounters no delay event and its data collection process takes one timeslot.

\textit{iii)} The worker encounters no delay event during the first hop. During the second hop, it encounters a delay event that costs one extra timeslot, and its data collection process takes one timeslot.

Thus, we define \( \mathcal{M}_n = \left\{ X, Y_1, Y_2, \ldots, Y_X , X', Z \right\} \) represent a TCS. Due to the randomness of uncertainties, each TCS can be described as the mobile worker encountering \( X \) delay events, with each delay event consuming \( Y_1, Y_2, \ldots, Y_X \) timeslots, and the number of times no delay events occur being \( X' \). The communication quality in this scenario is \( Z \). Due to \(\tau^{\text{delay}} \sim \mathbf{U}(t^{\text{min}}, t^{\text{max}})\), the value of $\text{Pr}(\tau^{\text{delay}}=\tau^{\prime\prime})$ is $\frac{1}{t^{\text{max}}-t^{\text{min}}+1}$, where $t^{\text{min}}\le \tau^{\prime\prime} \le t^{\text{max}} $. Similarly, the value of \(\text{Pr}(\gamma_{i,j}=\gamma_{i,j}^\prime)\) can be expressed as $\frac{1}{\mu_2-\mu_1+1}$, where $\mu_1 \le \gamma_{i,j}^\prime \le \mu_2 $. Therefore, we can calculate the probability of \( \beta_{i,j} \) as
	\begin{equation}{
			\begin{aligned}
				&\text{Pr}(\beta_{i,j}=1)=\\&	\sum_{\mathcal{M}_n\in\bm{\mathcal{C}}}\text{Pr}(\gamma_{i,j}=Z) (\text{Pr}(\alpha_{i,j}=1))^X(\text{Pr}(\alpha_{i,j}=1))^{X^\prime}\\&\times\prod_{\tau^{\prime\prime}=Y_1}^{Y_X}\text{Pr}(\tau^{\text{delay}}=\tau^{\prime\prime}) \\&=\sum_{\mathcal{M}_n\in\bm{\mathcal{C}}}\frac{(a_{i,j})^X(1-a_{i,j})^{X^\prime}}{(\mu_2-\mu_1+1)(t^{\text{max}}-t^{\text{min}}+1)^X}
		\end{aligned}}
	\end{equation}
\noindent\textbf{Derivation related to (10).}
The $ \text{E}\left[\beta_{i,j}\right] $, $ \text{E}\left[c_{i,j}\langle t_{i,j}^{\text{ini}} \rangle\right] $, and $ \text{E}\left[c_{i,j}^{\text{part}}\langle t_{i,j}^{\text{ini}} \rangle\right] $ of (10) are given by (37) and (38), respectively.

\noindent\textbf{Derivation related to (23c).}
In optimization problem $\bm{\mathcal{F}^W}$ given by (23), constraint (23c) represents a probabilistic expression, making its close form nontrivial to be obtained. To resolve such an issue, we transform (23c) into a tractable one by exploiting a set of bounding techniques. First, (23c) can be rewritten as
\begin{equation}
	R^{W}_1\left( w_j,s_i \right)\leq \rho_2 \Rightarrow \text{Pr}\left(U^W (w_j,s_i)\ge u_{\text{min}}\right) > 1-\rho_2.
\end{equation}
To obtain a tractable form for (39), we can have the upper-bound of its left-hand side by using Markov inequality, as the following (40).
\begin{equation}
	\text{Pr}\left(U^W (w_j,s_i)\ge u_{\text{min}}\right)\ge \frac{\text{E}[U^W (w_j,s_i)]}{u_{\text{min}}}
\end{equation}
where the value of $\text{E}[U^W (w_j,s_i)]$ is given by (10). 

\noindent\textbf{Derivation related to (23d).} Constraint (23d) can be rewritten as
\begin{equation}
	\begin{aligned}
			&R^{W}_2\left( w_j,s_i \right) =\text{Pr}\left(\beta_{i,j}=0\right)\leq \rho_3 \\&\Rightarrow \text{Pr}\left(\beta_{i,j}=1\right) > 1-\rho_3,
	\end{aligned}
\end{equation}
where $\text{Pr}\left(\beta_{i,j}=1\right)$ is given by (38).

\subsection{Derivations related to tasks}
\noindent\textbf{Mathematical expectation of $AGE_{i,j}$.}
\begin{equation}
	{\begin{aligned}
		\text{E}[AGE_{i,j}]&=\frac{\text{E}[age_{i,j}]}{\tau_{i,j}^{\text{sense}}+\text{E}[\tau_{i,j}^{\text{tran}}]}\\&=\frac{\sum_{t^\prime=t_{i,j}^{\text{gen}}}^{t_{i,j}^{\text{gen}}+\tau_{i,j}^{\text{sense}}+\text{E}[\tau_{i,j}^{\text{tran}}]}\left(t^\prime-t_{i,j}^{\text{gen}}\right)}{\tau_{i,j}^{\text{sense}}+\text{E}[\tau_{i,j}^{\text{tran}}]}\\&=\frac{\tau_{i,j}^{\text{sense}}+\text{E}[\tau_{i,j}^{\text{tran}}]+1}{2},
	\end{aligned}}
\end{equation}
where $\text{E}[\tau_{i,j}^{\text{tran}}]= \frac{d_i}{W\text{log}_2 \left(1+e_j^{\text{t}} \text{E}[\gamma_{i,j}]\right)}=\frac{d_i}{W\text{log}_2 \left(1+e_j^{\text{t}} \frac{\mu_1+\mu_2}{2}\right)}$.

\noindent\textbf{Derivation related to (15).}
The $ \text{E}[AGE_{i,j}] $ and $ \text{E}[\beta_{i,j}] $ of (15) are given by (38) and (42), respectively.

\noindent\textbf{Derivation related to (22c).}
In optimization problem $\bm{\mathcal{F}^S}$ given by (22), constraint (22c) represents a probabilistic expression, making its close form nontrivial to be obtained. To resolve such an issue, we transform (22c) into a tractable one by exploiting a set of bounding techniques. First, (22c) can be rewritten as
\begin{equation}
	R^{S}_1\left( w_j,s_i \right)\leq \rho_1 \Rightarrow \text{Pr}\left(Q(s_i,\varphi(s_i))\ge Q_i^{\text{D}}\right) > 1-\rho_1.
\end{equation}
To obtain a tractable form for (43), we can have the upper-bound of its left-hand side by using Markov inequality, as the following (44).
\begin{equation}
	\text{Pr}\left(Q(s_i,\varphi(s_i))\ge Q_i^{\text{D}}\right)\ge \frac{\text{E}[Q(s_i,\varphi(s_i))}{Q_i^{\text{D}}},
\end{equation}
where $\text{E}[Q(s_i,\varphi(s_i))]=\sum_{w_j\in \varphi(s_i)} \frac{1}{\text{E}[AGE_{i,j}]}$, and the value of $\text{E}[AGE_{i,j}]$ is given by (42). 

{\section{Computational Complexity Analysis of FT-SMP$^3$}
The primary computational complexity of our proposed FT-SMP$^3$ framework stems from the multi-round, many-to-many matching process conducted in the futures trading stage. In each round, workers dynamically update their task-path preferences, while tasks select optimal subsets of workers under budget constraints. The core computational components include:

\noindent \textit{i) Worker-side preference computation:} 
Before each matching round begins, each worker $w_j$ executes the EACO-P$^3$TR algorithm to search for a task sequence that maximizes its expected utility under spatiotemporal constraints. Let $|\bm{S}|$ denote the number of tasks, $|\bm{W}|$ is the number of workers, $\mathbb{K}$ refers to the number of ants per worker, and $iter_{\text{max}}$ represents the number of iterations. The complexity of this step per round is $\mathcal{O}(|\bm{W}| \cdot iter_{\text{max}} \cdot \mathbb{K} \cdot |\bm{S}|^2)$;

\noindent \textit{ii) Task-side worker selection:} 
Each task $s_i$ receives applications from a candidate set ${\widetilde{\mathbb{Y}}}(s_i)$ and solves a 0-1 knapsack problem to maximize its expected utility within budget $B_i$. The complexity per task is $\mathcal{O}(|{\widetilde{\mathbb{Y}}}(s_i)| \cdot B_i)$.
Summing over all tasks, the total cost per round is:
$\mathcal{O}\left( \sum_{s_i \in \bm{S}} |{\widetilde{\mathbb{Y}}}(s_i)| \cdot B_i \right)$; 

\noindent \textit{iii) Matching rounds:} 
The above processes are repeated for at most $\mathcal{X}$ rounds until convergence, i.e., when no further payment updates or selection changes occur.

\noindent \textbf{Overall Complexity:} 
The total time complexity of FT-SMP$^3$ is $\mathcal{O}\left(\mathcal{X} \cdot \left(|\bm{W}| \cdot iter_{\text{max}} \cdot \mathbb{K} \cdot |\bm{S}|^2 + \sum_{s_i \in \bm{S}} |{\widetilde{\mathbb{Y}}}(s_i)| \cdot B_i \right)\right)$.
This complexity is tractable under practical MCS system scales, especially considering that the futures matching stage operates offline and can leverage parallel and cloud computing resources.

\noindent \textbf{Worst-case Upper Bound:} 
In the worst-case scenario where every task receives applications from all workers (i.e., $|{\widetilde{\mathbb{Y}}}(s_i)| = |\bm{W}|$) and all tasks share a maximum budget $B_{\max} = \max\{B_i\},~\forall s_i \in \bm{S}$, the total complexity is bounded by:
$\mathcal{O}\left(\mathcal{X} \cdot |\bm{W}| \cdot \left(iter_{\text{max}} \cdot \mathbb{K} \cdot |\bm{S}|^2 + |\bm{S}| \cdot B_{\max} \right)\right)$.}

\section{Property Analysis on FT-M2M Matching of FT-SMP$^3$}
We next examine the aforementioned property of FT-M2M matching, as outlined below
\begin{lem}
	(Convergence of FT-M2M matching) Alg. 2 converges within finite rounds.
\end{lem}
\begin{proof}
	We utilize DP algorithm to transform the problem into a two-dimensional 0-1 knapsack problem as shown in Alg. 2. After a finite number of rounds, each worker's asked payment can either be accepted or the risk of obtaining expected utility will be unacceptable, supporting the property of convergence.
\end{proof}

\begin{lem}(Individual rationality of FT-M2M matching) Our proposed M2M matching satisfies the individual rationality of all tasks and workers. \end{lem}
\begin{proof}
	The individual rationality of each task and worker is proved respectively, as the following: 
	
\noindent \textbf{Individual rationality of tasks.} For each task $ s_i\in\bm{S} $, since the designed 0-1 knapsack problem regards $ B_i $ as the corresponding capacity, the overall payment of $ s_i $ will thus not exceed $ B_i $. Moreover, thanks to the factor of risk analysis and control of possible risk, e.g., constraint (22c), each task $ s_i $ can decide whether to sign long-term contracts with the matched workers under an acceptable risk, which thus ensures that the desired service quality of each task, at a high probability. 
	
\noindent \textbf{Individual rationality of workers.} Thanks to stringent risk control measures, each worker thoroughly evaluates the risks on failing to complete task and obtaining undesired expected utility before sending requests to tasks on their preference list. These risks are kept within reasonable range; otherwise, the worker will not sign long-term contracts with these tasks.
	
	As a result, our proposed M2M matching in the futures market is individual rational.
\end{proof}

\begin{lem} (Fairness of FT-M2M matching) Our proposed FT-M2M matching guarantees fairness in the futures market. \end{lem}
\begin{proof}
	According to Definition 4, fairness indicates the case without type 1 blocking coalition, we offer the proof of Lemma 3 by contradiction. 
	
	Under a given matching $ \varphi $, worker $ w_j $ and task set $ \mathbb{S} $ can form a type 1 blocking coalition $ (w_j; \mathbb{S}) $, as shown by (18) and (19). If task $ s_i $ does not sign a long-term contract with worker $ w_j $, the payment of worker $ w_j $ during the last round can only be the cost, as given by (45) and (46).
	\begin{equation}\label{key}
		\begin{aligned}
			p_{i,j}\left\langle k \right\rangle = { c}_{i,j}\langle t^{\text{ini}} \rangle,
		\end{aligned}
	\end{equation}
	\begin{equation}\label{key}
		\begin{aligned}
			\overline{U^{S}}\left({ s_{i},\left\{\varphi\left( s_{i} \right)\backslash\varphi^{\prime}\left( s_{i} \right)\right\} \cup \left\{ w_{j} \right\} } \right) < \overline{U^{S}}\left( t_{i},\varphi\left( s_{i} \right) \right).
		\end{aligned}
	\end{equation}
	If task $ s_i $ selects worker $ w_j $, we have $ p_{i,j}\left\langle k^{*} \right\rangle \geq p_{i,j}\left\langle k \right\rangle = {c}_{i,j} $ and the following (47)
	\begin{equation}\label{key}
		\begin{aligned}
			&\overline{U^{S}}\left(s_{i},\left\{\varphi\left( s_{i} \right)\backslash\varphi^{\prime}\left( s_{i} \right)\right\} \cup \left\{ w_{j} \right\} \right) \geq\\& \overline{U^{S}}\left( s_{i},\left\{\varphi\left( s_{i} \right)\backslash\varphi^{\prime\prime}\left( s_{i} \right)\right\} \cup \left\{ w_{j} \right\} \right),\\
		\end{aligned}
	\end{equation}
where $ 
\varphi^{\prime\prime}\left( s_{i} \right) \subseteq \varphi^{\prime}\left( s_{i} \right) $. From (46) and (47), we can get
\begin{equation}\label{key}
		\begin{aligned}
				\overline{U^{S}}\left( {s_{i},\varphi\left( s_{i} \right) } \right) > \overline{U^{S}}\left(s_{i},\left\{\varphi\left( s_{i} \right)\backslash\varphi^{\prime\prime}\left( s_{i} \right)\right\} \cup \left\{ w_{j} \right\} \right),
			\end{aligned}
	\end{equation}
which is contrary to (19). Thus, our proposed FT-M2M matching ensures the property of fairness. 
\end{proof}

\noindent
\begin{lem} (Non-wastefulness of FT-M2M matching) Alg. 2 satisfies the property of non-wastefulness. \end{lem}
\begin{proof}
	We conduct the proof of Lemma 4 by contradiction. Under a given matching $ \varphi $, worker $ w_j $ and task set $ \mathbb{S} $ form a type 2 blocking coalition $ (w_j; \mathbb{S}) $, as shown by (20) and (21). 
	
	If task $ s_i $ rejects $ w_j $, the payment of $ w_j $ during the last round can only be $ p_{i,j}^F\left\langle k \right\rangle = {c}_{i,j}\langle t^{\text{ini}}_{i,j}\rangle $, where the only reason of the rejection between $ s_i $ and $ w_j $ is the overall payment exceeds the limited budget $ B_i $. However, the coexistence of (20) and (21) shows that task $ s_i $ has an adequate budget to recruit workers, which contradicts the aforementioned assumption. Therefore, our proposed FT-M2M matching in the futures market is non-wasteful.
\end{proof}

\begin{thm}(Strong stability of FT-M2M matching) FT-M2M matching is strongly stable. \end{thm}
\begin{proof}
	Since the matching result of Alg. 1 holds Lemma 2, Lemma 3, and Lemma 4, according to Definition 6, our proposed FT-M2M matching is strongly stable.
\end{proof}

\begin{thm}
	(Competitive equilibrium associated with service trading between workers and tasks in FT-SMP$^3$) The trading between workers and tasks can reach a competitive equilibrium.
\end{thm}
\begin{proof}
	To prove this theorem, we discuss that the three conditions introduced by Definition 7 can be held in workers-tasks trading. First, we set $ p_{i,j}^F \geq \text{E}\left[c_{i,j}\langle t^{\text{ini}}_{i,j}\rangle\right] $, indicating that the expected service cost will be covered by its asked payment in each round (e.g., constraint (23b)). We next demonstrate that when a task $ s_i $ enters into a long-term contract with workers $ w_j $, task $ s_i $ achieves maximum expected utility. This is attributed to the fact that $ s_i $ selects the workers based on DP algorithm (e.g., line 10, Alg. 2), ensuring the attainment of the maximum expected utility for $ s_i $. Then, if $ s_i $ is not matched to more workers $ w_j \in \bm{W} $, its remaining budget will not recruit additional worker (e.g., proof of Lemma 4). 
	
	According to Definition 7, we can verify that the considered worker-task trading in futures market can reach a competitive equilibrium.
\end{proof}

\begin{thm}
	(Weak Pareto optimality associated with service trading between workers and tasks in FT-SMP$^3$) The proposed associated with service trading between workers and tasks in FT-SMP$^3$ provides a weak Pareto optimality.
\end{thm}
\begin{proof}
Recall the design of FT-M2M matching, each worker makes decisions based on their preference list. In particular, this preference list is determined by our proposed EACO-P$^3$TR algorithm, accounting for the diverse demands of tasks (e.g., time windows and locations), along with the asked payment and costs associated with each worker. This enables that the selected task vector can maximize the expected utility for the worker. For each task owner \( s_i \), if a worker \( w_j \) can offer a higher expected utility compared to its current matches, \( s_i \) prefers establishing a new matching. However, this could potentially create a blocking pair. Theorem 1 confirms that our matching is stable and allows no blocking pairs. Thus, there exists no Pareto improvements, making the service trading weakly Pareto optimal.
\end{proof}

\section{Derivations Associated with ST-DP$^2$WR Mechanism}
\noindent\textbf{Derivation related to (33c).}
In optimization problem $\bm{\mathcal{F}^{W\prime}}$ given by (33), constraint (33c) represents a probabilistic expression, making its close form nontrivial to be obtained. To resolve this issue, we transform (33c) into a tractable one by exploiting a set of bounding techniques. First, (33c) can be rewritten as
\begin{equation}
	R^{W\prime}_1\left( w_j,s_i \right)\leq \rho_4 \Rightarrow \text{Pr}\left(U^{W\prime} (w_j,s_i)\ge u_{\text{min}}\right) > 1-\rho_4.
\end{equation}
To obtain a tractable form for (49), we can have the upper-bound of its left-hand side by using Markov inequality, as the following (50).
\begin{equation}
	\text{Pr}\left(U^{W\prime} (w_j,s_i)\ge u_{\text{min}}\right)\ge \frac{\text{E}[U^{W\prime} (w_j,s_i)]}{u_{\text{min}}}
\end{equation}
where the value of $\text{E}[U^{W\prime} (w_j,s_i)]$ is expressed as
\begin{equation}
	\begin{aligned}
	&\text{E}[U^{W\prime} (w_j,s_i)]=\sum_{s_i\in\nu_t(w_j)}\text{E}[\beta_{i,j}^{S}]\left(p_{i,j}^{S}-\text{E}[c_{i,j}\langle t_{i,j}^{\text{ini}} \rangle]\right)\\&-\sum_{s_i\in\nu_t(w_j)}(1-\text{E}[\beta_{i,j}^{S}])\left(\text{E}[c_{i,j}^{\text{part}}\langle t_{i,j}^{\text{ini}} \rangle]+q_{i,j}^{S}\right),
\end{aligned}
\end{equation}
where $\text{E}[c_{i,j}\langle t_{i,j}^{\text{ini}} \rangle]$ and $ \text{E}[c_{i,j}^{\text{part}}\langle t_{i,j}^{\text{ini}} \rangle]$ are given by (37) and (38), respectively. Besides, $\text{E}[\beta_{i,j}^{S}]$ can be calculated as $\text{E}[\beta_{i,j}^{S}]=\text{Pr}(\text{E}[\beta_{i,j}^{S}]=1)\times 1+ \text{Pr}(\text{E}[\beta_{i,j}^{S}]=0)\times 0 = \text{Pr}(\text{E}[\beta_{i,j}^{S}]=1)$.
We use \( \bm{\mathcal{C}}^\prime \) to denote the set of possible TCS, and \( \mathcal{M}^\prime_n = \left\{ X, Y_1, Y_2, \ldots, Y_X , X', Z \right\} \) to represent a TCS. Due to the randomness of uncertainties, each task completion scenario can be described as the mobile worker encountering \( X \) delay events, with each delay event consuming \( Y_1, Y_2, \ldots, Y_X \) timeslots, and the number of times no delay events occur being \( X' \). The communication quality in this scenario is \( Z \). Due to \(\tau^{\text{delay}} \sim \mathbf{U}(t^{\text{min}}, t^{\text{max}})\), the value of $\text{Pr}(\tau^{\text{delay}}=\tau^{\prime\prime})$ is $\frac{1}{t^{\text{max}}-t^{\text{min}}+1}$, where $t^{\text{min}}\le \tau^{\prime\prime} \le t^{\text{max}} $. Similarly, the value of \(\text{Pr}(\gamma_{i,j}=\gamma_{i,j}^\prime)\) can be expressed as $\frac{1}{\mu_2-\mu_1}$, where $\mu_1 \le \gamma_{i,j}^\prime \le \mu_2 $. Therefore, we can calculate the probability of \( \beta_{i,j}^\prime \) as
\begin{equation}{
		\begin{aligned}
			&\text{Pr}(\beta^\prime_{i,j}=1)=\\&	\sum_{\mathcal{M}^\prime_n\in\bm{\mathcal{C}^\prime}}\text{Pr}(\gamma_{i,j}=Z) (\text{Pr}(\alpha_{i,j}=1))^X(\text{Pr}(\alpha_{i,j}=1))^{X^\prime}\\&\times\prod_{\tau^{\prime\prime}=Y_1}^{Y_X}\text{Pr}(\tau^{\text{delay}}=\tau^{\prime\prime}) \\&=\sum_{\mathcal{M}^\prime_n\in\bm{\mathcal{C}}^\prime}\frac{(a_{i,j})^X(1-a_{i,j})^{X^\prime}}{(\mu_2-\mu_1+1)(t^{\text{max}}-t^{\text{min}}+1)^X}
	\end{aligned}}
\end{equation}

\noindent\textbf{Derivation related to (33d).} Constraint (33d) can be rewritten as
\begin{equation}
	\begin{aligned}
		&R^{W\prime}_2\left( w_j,s_i \right) =\text{Pr}\left(\beta_{i,j}^S=0\right)\leq \rho_5 \\&\Rightarrow \text{Pr}\left(\beta_{i,j}=1\right) > 1-\rho_5,
	\end{aligned}
\end{equation}
where $\text{Pr}\left(\beta_{i,j}^S=1\right)$ is given by (52).

\section{Detail of ST-M2M Matching}
\subsection{Key Definitions of Matching}
 We next describe the basic characteristics of this ST-M2M matching.

\begin{Defn}(ST-M2M matching)
	A M2M matching $ \nu_t $ of ST-DP$^2$WR mechanism constitutes a mapping between task set $ \bm{S^\prime}\langle t\rangle $ and worker set $ \bm{W^\prime}\langle t\rangle $, which satisfies the following properties:
	
	\noindent
	$\bullet$ for each task $ s_{i} \in \bm{S^\prime}\langle t\rangle,\nu_t\left( s_i \right) \subseteq \bm{W^{\prime}}\langle t\rangle $,
	
	\noindent
	$\bullet$ for each worker $ w_{j} \in \bm{W^{\prime}}\langle t\rangle, \nu_t\left( w_{j} \right) \subseteq \bm{S^\prime}\langle t\rangle $,
	
	\noindent
	$\bullet$ for each task $ s_i $ and worker $ w_j $, $ s_i\in\nu_t(w_j)$ if and only if $ w_j\in\nu_t\left(s_i\right) $.
\end{Defn}

We next define a concept called \textit{blocking coalition}, which is a significant factor that may make the matching unstable.

\begin{Defn}(Blocking coalition)
	Under a given matching $ \nu_t $, worker $ w_j $ and task set $ \mathbb{S}^\prime\langle t\rangle \subseteq \bm{S^\prime}\langle t\rangle$ may form one of the following two types of blocking coalition $ \left(w_j; \mathbb{S}^\prime\langle t\rangle\right) $.
	
	\textbf{Type 1 blocking coalition:} Type 1 blocking coalition satisfies the following two conditions:
	
	\noindent
	$\bullet$ Worker $ w_j $ prefers execution of task set $ \mathbb{S}^\prime\langle t\rangle \subseteq \bm{S^\prime}\langle t\rangle $ to its currently matched task set $ \nu_t(w_j) $, i.e., 
	\begin{equation}\label{key}
		\begin{aligned}
			{U^{W\prime}_t}(w_j,\mathbb{S}^\prime\langle t\rangle)>{U^{W\prime}_t}(w_j,\nu_t(w_j)). 
		\end{aligned} 
	\end{equation}

	\noindent
	$\bullet$ Every task in $ \mathbb{S}^\prime\langle t\rangle $ prefers to recruit workers rather than being matched to its currently matched/assigned worker set. That is, for any task $ s_i\in \mathbb{S}^\prime\langle t\rangle $, there exists worker set $ \nu_t^\prime(s_i) $ that constitutes the workers that need to be evicted, satisfying
	\begin{equation}\label{key}
		\begin{aligned}
			{U^{S\prime}_t}\left(s_i,\left\{\nu_t\left( s_i \right)\backslash\nu_t^{\prime}\left( s_i \right)\right\} \cup \left\{ w_{j} \right\} \right) > {U^{S\prime}_t}\left(s_i,\nu_t\left( s_i \right) \right).\\
		\end{aligned}
	\end{equation} 
	
	\textbf{Type 2 blocking coalition:} Type 2 blocking coalition satisfies the following two conditions:
	
	\noindent
	$\bullet$ Worker $ w_j $ prefers executing task set $ \mathbb{S}^\prime\langle t\rangle \subseteq \bm{S^\prime}\langle t\rangle $ to its currently matched task set $ \nu_t(w_j) $, i.e.,
	\begin{equation}\label{key}
		\begin{aligned}
			{U^{W\prime}_t}(w_j,\mathbb{S}^\prime\langle t\rangle )>{U^{W\prime}_t}(w_j,\nu_t(w_j) ).
		\end{aligned}
	\end{equation} 
	
	\noindent
	$\bullet$ Every task in $ \mathbb{S}^\prime\langle t\rangle $ prefers to further recruit worker $ w_j $ in conjunction to its currently matched/assigned worker set. That is, for any task $ s_i\in \mathbb{S}^\prime\langle t\rangle $, we have
	\begin{equation}\label{key}
		\begin{aligned}
			{U^{S\prime}_t}(s_i,\nu_t(s_i)\cup\left\{ w_{j} \right\})>{U^{S\prime}_t}(s_i,\nu_t(s_i) ) .
		\end{aligned} 
	\end{equation}
\end{Defn}

\subsection{Algorithm Analysis}
\begin{algorithm}[] 
	{\small\caption{{Proposed ST-M2M Matching Algorithm}}
		\LinesNumbered 
		\textbf{Initialization: $ k \leftarrow 1 $, $ p_{i,j}^S\left\langle 1 \right\rangle \leftarrow p^{\text{Desire}}_{i,j}$, for $ \forall i,j $, $ {flag}_{j} \leftarrow 1 $, $\mathbb{V}\left( w_{j} \right)\leftarrow \varnothing$, $\mathbb{V}\left( s_{i} \right)\leftarrow \varnothing$}\ 

		\While{$ {flag}_{j} $}{
			\textbf{$ {flag}_{j} \leftarrow 0 $}
			
			$ \mathbb{V}\left( w_{j} \right) \leftarrow $ selected $s_i$ from $\bm{S^{\prime}}\langle t\rangle$ accroding to (34c) and (34d)
			
			\If{$ \forall\mathbb{V}\left( w_{j} \right) \neq \varnothing $}{
				\For{$ \forall s_i \in \mathbb{V}\left( w_{j} \right) $}{$ w_j $ sends a proposal about its information to $ s_i $}
				
				\While{
					$ \Sigma_{w_{j}\in \bm{W^{\prime}}\langle t\rangle}{flag}_{j} > 0 $}{
					Collect proposals from the workers in $ \bm{W^{\prime}}\langle t\rangle $, e.g., using $ {\widetilde{\mathbb{V}}}\left(s_i\right) $ to include the workers that send proposals to $ s_i $
					
					$ \mathbb{V}(s_i)\leftarrow $ choose workers from $ {\widetilde{\mathbb{V}}}\left(s_i\right) $ that can achieve the maximization of the expected utility by using DP under budget $ B_i^t $
					
					$ s_i $ temporally accepts the workers in $ \mathbb{V}(s_i) $, and rejects the others
				}
				
				\For{
					$ \forall s_i \in \mathbb{V}\left( w_{j} \right) $
				}{
					\If{$ w_j $ is rejected by $ s_i $ and $ p_{i,j}^S\left\langle k \right\rangle>c_{i,j}\langle t^{\text{ini}}_{i,j}\rangle $}{
						$ p_{i,j}^S\left\langle {k + 1} \right\rangle \leftarrow \max\left\{ p_{i,j}^S\left\langle k \right\rangle - \mathrm{\Delta}p~,{ c}_{i,j}\langle t^{\text{ini}}_{i,j}\rangle \right\} $}
					\Else{$ p_{i,j}^S\left\langle {k + 1} \right\rangle \leftarrow p_{i,j}^S\left\langle k \right\rangle $}
				}
				
				\If{there exists $p_{i,j}^S\left\langle {k + 1} \right\rangle \neq p_{i,j}^S\left\langle k \right\rangle\ $, $ \forall s_i\in\mathbb{V}\left( w_{j} \right) $}{
					$ {flag}_j\leftarrow 1 $,	$ k\leftarrow k+1 $\
				}
			}
		}

		$\nu_t(s_i)\leftarrow\mathbb{V}(s_i)$, $\nu_t(w_j)\leftarrow \mathbb{V}(w_j)$
		
		\textbf{Return:} $\nu_t(s_i)$, $\nu_t(w_j)$}
\end{algorithm}
We next describe the steps of ST-M2M matching, the pseudo-code of which is given in Alg. 4.

\noindent
\textbf{Step 1. Initialization:} At the beginning of Alg. 4, each worker $ w_j $'s asked payment is set to $ p_{i,j}^S\left\langle 1 \right\rangle = p^{\text{Desire}}_{i,j} $ (line 1, Alg. 4), $ \mathbb{V} (w_j ) $ contains the interested tasks of $ w_j $ and $ \mathbb{V}(s_i) $ involves the workers temporarily selected by $ s_i $.

\noindent
\textbf{Step 2. Proposal of workers:} At each round $ k $, each worker $ w_j $ first chooses tasks from $\bm{S^{\prime}}\langle t\rangle$ according to its risk management (33c) and (33d), and records them in $ \mathbb{V}(w_j ) $. Then, $ w_j $ sends a proposal to each task in $ \mathbb{V}(w_j) $, including its asked payments $ p_{i,j}^S\left\langle k \right\rangle $, probability of completing $s_i$ (i.e., $\beta^S_{i,j}$), and expected service quality of sensing data $\text{E}[Q_{i,j}]$ (line 7, Alg. 4).

\noindent
\textbf{Step 3. Worker selection on tasks' side:} After collecting the information from workers in set $ {\widetilde{\mathbb{V}}}\left(s_i\right) $, each task $ s_i $ solves a 0-1 knapsack problem,
which can generally be solved via DP (line 10, Alg. 4), determine a collection of temporary workers (e.g., set $ \mathbb{V} (s_i)$), where $\mathbb{V} (s_i)\subseteq {\widetilde{\mathbb{V}}}\left(s_i\right) $ that can bring the maximum expected utility, under budget $ B_i^t $. Then, each $ s_i $ reports its decision on worker selection during the current round to all the workers. 

\noindent
\textbf{Step 4. Decision-making on workers' side:} After obtaining the decisions from each task $ s_i\in\mathbb{V} (w_j) $, worker $ w_j $ inspects the following conditions:

\noindent
$\bullet$ \textbf{Condition 1.} If $ w_j $ is accepted by $ s_i $ or its current asked payment $ p_{i,j}^S\left\langle k \right\rangle $ equals to its cost $ {c}_{i,j}\langle t^{\text{ini}}_{i,j}\rangle $, its payment remains unchanged (line 16, Alg. 4);

\noindent
$\bullet$ \textbf{Condition 2.} If $ w_j $ is rejected by a task $ s_i $ and its asked payment $ p_{i,j}^S\left\langle k \right\rangle $ can still cover its cost $ {c}_{i,j}\langle t^{\text{ini}}_{i,j}\rangle $, it decreases its asked payment for $ s_i $ in the next round, while avoiding a negative utility, as follows (line 14, Alg. 4):
\begin{equation}\label{key}	
	{
		\begin{aligned}
			p_{i,j}^S\left\langle {k + 1} \right\rangle = \max\left\{ p_{i,j}^S\left\langle k \right\rangle - \mathrm{\Delta}p_j~,{c}_{i,j}\langle t^{\text{ini}}_{i,j}\rangle \right\}.
	\end{aligned} }
\end{equation}

\noindent
\textbf{Step 5. Repeat:} If all the asked payments stay unchanged from the $ (k-1)^{\text{th}} $ round to the $ k^{\text{th}} $ round, the matching will be terminated at round $ k $. We use $ \Sigma_{w_j\in \bm{W^{\prime}}\langle t\rangle}{flag}_j=0 $ to denote this situation (line 3, Alg. 4).
Otherwise, the algorithm repeats the above steps (e.g., lines 2-18, Alg. 4) in the next round.

\subsection{Design Targets}
Desired properties of ST-M2M matching are detailed below.

\begin{Defn}(Individual rationality of ST-M2M matching) For both tasks and workers, a matching $ \nu_t $ is individually rational when the following conditions are satisfied:
	
	\noindent
	$\bullet$ For tasks: the overall payment of a task $s_i$ matched to workers $\nu_t\left(s_i\right)$ does not exceed $B_i^t$, i.e., constraint (33b) is met.
	
	\noindent
	$\bullet$ For workers: \textit{i)} the risk of each worker receiving an undesired utility is controlled within a certain range, i.e., constraint (34c) is satisfied; \textit{ii)} the risk of each worker failing to complete matched tasks is controlled within a certain range, i.e., constraint (34d) is satisfied.
\end{Defn}

\begin{Defn}(Fairness of ST-M2M matching): A matching $\nu_t$ is fair if and only if it imposes no type 1 blocking coalition.\end{Defn}
\begin{Defn}(Non-wastefulness of ST-M2M matching): A matching $\nu_t$ is non-wasteful if and only if it imposes no type 2 blocking coalition. \end{Defn}

\begin{Defn}(Strong stability of ST-M2M matching) The proposed ST-M2M matching is strongly stable if it is individual rationality, fair, and non-wasteful.
\end{Defn}

\begin{Defn}(Competitive equilibrium associated with service trading between workers and task owners in ST-DP$^2$WR) The trading between workers and task owners reaches a competitive equilibrium if the following conditions are satisfied:
	
	\noindent
	$\bullet$ For each worker $ w_j \in \bm{W^{\prime}}\langle t\rangle $, if $ w_j $ is associated with a task owner $ s_i\in \bm{S^\prime}\langle t\rangle $, then $ c_{i,j}\langle t^{\text{ini}}_{i,j}\rangle \leq p^S_{i,j} $;
	
	\noindent
	$\bullet$ For each task $ s_i\in \bm{S^\prime}\langle t\rangle $, $ s_i $ is willing to trade with the worker that can bring it with the maximum utility;
	
	\noindent
	$\bullet$ For each task \( s_i \) in set \( \bm{S^\prime}\langle t\rangle \), if \( s_i \) does not recruit more workers, then the remaining budget after deducting the payments made to workers \( w_j \) is insufficient to recruit an additional worker.
\end{Defn}

For a MOO problem (e.g., optimization problems $ \bm{\mathcal{F}^{W\prime}} $ and $ \bm{\mathcal{F}^{S\prime}} $), a Pareto improvement occurs when the \textit{social welfare} can be increased with another feasible matching result\cite{pareto 1}. Specifically, the social welfare refers to a summation of utilities of workers and task owners in our designed matching. Thus, a matching is weak Pareto optimal when there is no Pareto improvement.

\begin{Defn}(Weak Pareto optimality of trading between tasks and workers in ST-DP$^2$WR) The proposed matching game is weak Pareto optimal if there is no Pareto improvement.
\end{Defn}

\subsection{Proof of Matching Property}
ST-M2M matching satisfies the properties of convergence, individual rationality, fairness, non-wastefulness, strong stability, competitive equilibrium, and weak Pareto optimality.

\begin{lem}
	(Convergence of ST-M2M matching) Alg. 4 converges within finite rounds.
\end{lem}
\begin{proof}
	We utilize DP algorithm to transform the problem into a two-dimensional 0-1 knapsack problem of Alg. 4. After a finite number of rounds, each worker's asked payment can either be accepted or the risk of obtaining utility will be unacceptable, supporting the property of convergence.
\end{proof}

\begin{lem}(Individual rationality of ST-M2M matching) Our proposed M2M matching satisfies the individual rationality of all tasks and workers. \end{lem}
\begin{proof}
The individual rationality of each task and worker is proved respectively, as the following: 

\noindent \textbf{Individual rationality of tasks.} For each task $ s_i\in\bm{S^\prime}\langle t\rangle $, since the designed 0-1 knapsack problem regards $ B_i^t $ as the corresponding capacity, the overall payment of $ s_i $ will thus not exceed $ B_i^t $. 

\noindent \textbf{Individual rationality of workers.} Thanks to stringent risk control measures, each worker thoroughly evaluates the risks on failing to complete task and obtaining undesired utility before sending requests to tasks on their preferred list. These risks are kept within reasonable limits; otherwise, the worker will not choose to service with the task.

As a result, our proposed M2M matching in the spot trading stage is individual rational.
\end{proof}

\begin{lem} (Fairness of ST-M2M matching) Our proposed ST-M2M matching guarantees fairness in the futures trading stage. \end{lem}
\begin{proof}
	According to Definition 12, fairness indicates the case without type 1 blocking coalition, we offer the proof of Lemma 7 by contradiction. 
	
	Under a given matching $ \nu_t $, worker $ w_j $ and task set $ \mathbb{S}^\prime\langle t\rangle $ can form a type 1 blocking coalition $ (w_j; \mathbb{S}^\prime\langle t\rangle) $, as shown by (54) and (55).
	If task $ s_i $ does not sign a long-term contract with worker $ w_j $, the payment of worker $ w_j $ during the last round can only be the cost, as given by (59) and (60).
	\begin{equation}\label{key}
		\begin{aligned}
			p_{i,j}^S\left\langle k \right\rangle = { c}_{i,j}\left\langle t^{\text{ini}}_{i,j} \right\rangle,
		\end{aligned}
	\end{equation}
	\begin{equation}\label{key}
		\begin{aligned}
			{U^{S\prime}_t}\left({ s_{i},\left\{\nu_t\left( s_{i} \right)\backslash\nu_t^{\prime}\left( s_{i} \right)\right\} \cup \left\{ w_{j} \right\} } \right) < {U^{S\prime}_t}\left( t_{i},\nu_t\left( s_{i} \right) \right).
		\end{aligned}
	\end{equation}
	If task $ s_i $ selects worker $ w_j $, we have $ p_{i,j}\left\langle k^{*} \right\rangle \geq p_{i,j}\left\langle k \right\rangle = {c}_{i,j} $ and the following (61)
	\begin{equation}\label{key}
		\begin{aligned}
			&{U^{S\prime}_t}\left(s_{i},\left\{\nu_t\left( s_{i} \right)\backslash\nu_t^{\prime}\left( s_{i} \right)\right\} \cup \left\{ w_{j} \right\} \right) \geq\\& {U^{S\prime}_t}\left( s_{i},\left\{\nu_t\left( s_{i} \right)\backslash\nu_t^{\prime\prime}\left( s_{i} \right)\right\} \cup \left\{ w_{j} \right\} \right),\\
		\end{aligned}
	\end{equation}
	where $ 
	\nu_t^{\prime\prime}\left( s_{i} \right) \subseteq \nu_t^{\prime}\left( s_{i} \right) $. From (60) and (61), we can get
	\begin{equation}\label{key}
		\begin{aligned}
			{U^{S\prime}_t}\left( {s_{i},\nu_t\left( s_{i} \right) } \right) > {U^{S\prime}_t}\left(s_{i},\left\{\nu_t\left( s_{i} \right)\backslash\nu_t^{\prime\prime}\left( s_{i} \right)\right\} \cup \left\{ w_{j} \right\} \right),
		\end{aligned}
	\end{equation}
	which is contrary to (55). Thus, our proposed ST-M2M matching ensures the property of fairness. 
\end{proof}

\noindent
\begin{lem} (Non-wastefulness of ST-M2M matching) Alg. 4 satisfies the property of non-wastefulness. \end{lem}
\begin{proof}
	We conduct the proof of Lemma 8 by contradiction. Under a given matching $ \nu_t $, worker $ w_j $ and task set $ \mathbb{S}^\prime\langle t\rangle $ form a type 2 blocking coalition $ (w_j; \mathbb{S}^\prime\langle t\rangle) $, as shown by (56) and (57). 
	
	If task $ s_i $ rejects $ w_j $, the payment of $ w_j $ during the last round can only be $ p_{i,j}^S\left\langle k \right\rangle = {c}_{i,j} $, where the only reason of the rejection between $ s_i $ and $ w_j $ is the overall payment exceeds the limited budget $ B_i $. However, the coexistence of (56) and (57) shows that task $ s_i $ has an adequate budget to recruit workers, which contradicts the aforementioned assumption. Therefore, our proposed M2M matching in the spot trading stage is non-wasteful.
\end{proof}

\begin{thm}(Strong stability of ST-M2M matching) ST-M2M matching is strongly stable. \end{thm}
\begin{proof}
	Since the matching results of Alg. 4 holds Lemma 6, Lemma 7, and Lemma 8, according to Definition 14, our proposed ST-M2M matching is strongly stable.
\end{proof}

\begin{thm}
	(Competitive equilibrium associated with resource trading between workers and tasks) The trading between workers and tasks can reach a competitive equilibrium.
\end{thm}
\begin{proof}
To prove this theorem, we discuss that the three conditions introduced by Definition 15 can be held. First, we set $ c_{i,j}\langle t^{\text{ini}}_{i,j}\rangle \leq p^S_{i,j} $, indicating that the service cost will be covered by its asked payment in each round (e.g., constraint (23b)). We next demonstrate that when a task $ s_i $ recruits workers, $ s_i $ achieves maximum expected utility. This is attributed to the fact that $ s_i $ selects the workrs based on DP algorithm (e.g., line 10, Alg. 4), ensuring the attainment of the maximum utility for $ s_i $. Then, if $ s_i $ is not matched to more workers $ w_j \in \bm{W^{\prime}}\langle t\rangle $, its remaining budget will not recruit additional worker (e.g., proof of Lemma 8). 

According to Definition 15, we can verify that the considered worker-task trading in futures market can reach a competitive equilibrium.
\end{proof}

\begin{thm}
	(Weak Pareto optimality of ST-M2M matching) The proposed ST-M2M matching provides a weak Pareto optimality.
\end{thm}
\begin{proof}
Recall the design of ST-M2M matching, worker $w_j$ evaluates whether accepting a new task $s_i \in \bm{S^{\prime}_t}\langle t\rangle$ can affect the execution of subsequent tasks in $ \nu_t(w_j)$, while ensuring the risks associated with these new tasks are within reasonable bounds. Based on these assessments, a worker accepts new tasks only if doing so can bring higher utility. For each task owner \( s_i \), if there exists a worker \( w_j \) who can offer a higher utility than the currently matched workers, \( s_i \) is more likely to form a new matching relationship, even though this may result in a blocking pair. According to Theorem 4, our proposed matching is stable and free of blocking pairs. As a result, when our proposed matching procedure concludes, no Pareto improvements are possible, making the matching weakly Pareto optimal.
\end{proof}

\section{Supplementary Performance Evaluations}
\subsection{Individual Rationality of Tasks and Workers}
\begin{figure}[htb] \centering 
	\vspace{-1.9999cm}
	\subfigtopskip=2pt
	\subfigbottomskip=10pt
	\subfigcapskip=-2.0cm
	\setlength{\abovecaptionskip}{-1.6cm}
	\subfigure[] {
		\label{fig:a} 
		\includegraphics[width=0.50\columnwidth]{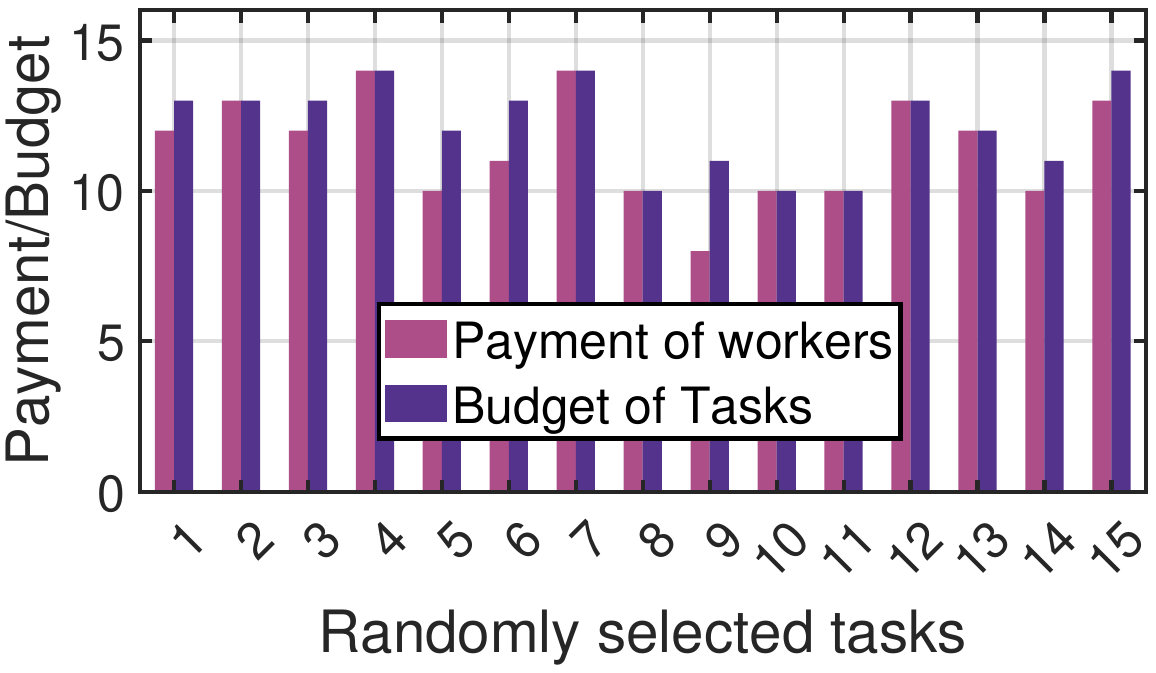} 
	} \hspace{-5.8mm}
	\subfigure[] { 
		\label{fig:b} 
		\includegraphics[width=0.50\columnwidth]{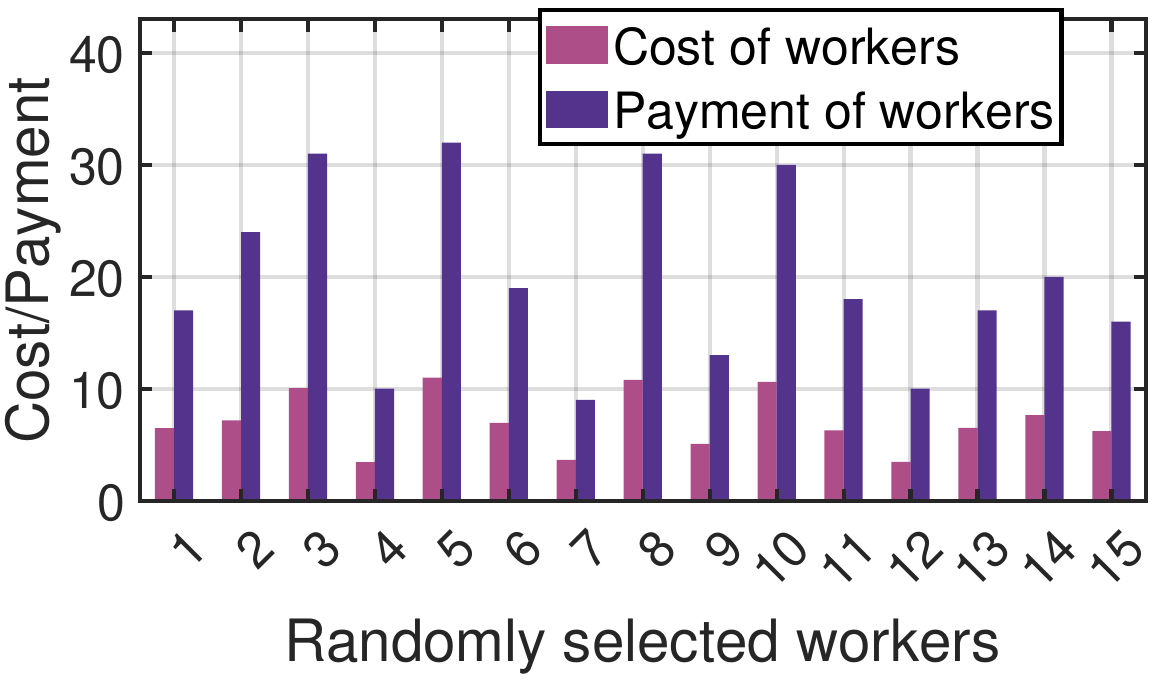} 
	}  
	\caption{Performance comparisons in terms of individual rationality.} 
	\label{fig} 
	\vspace{-0.2 cm}
\end{figure} 
To demonstrate that our StagewiseTM$^3$atch can support individual rationality, we randomly select 15 tasks (among 30 ones) and show their budgets along with the payments in Fig. 7(a). As we can see from Fig. 7(a), the total expenses of tasks that should be paid to workers consistently stay within their budgets, confirming that the set of matching mechanisms of StagewiseTM$^3$atch method uphold the individual rationality of task owners. Also, Fig. 7(b) considers 15 randomly selected workers out of 30 ones, and shows the total service cost incurred on each worker never surpasses the total payments they receive, accordingly. Highlighting the individual rationality of workers.

{\subsection{Average Service Quality, Utility}
\begin{figure*}[thb]{
	\centering
	\setlength{\abovecaptionskip}{-1 mm}
	\includegraphics[width=2\columnwidth]{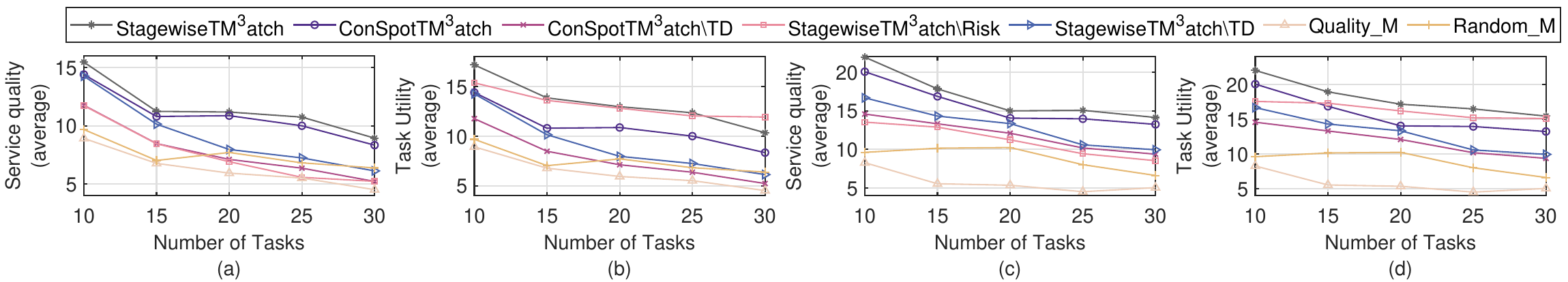}
	\caption{Performance comparison in terms of average service quality and average task utility under different problem sizes, where (a)-(b) consider 15 workers, and (c)-(d) consider 30 workers in the network.}
	\vspace{-0.5cm}}
\end{figure*}
To provide a more comprehensive evaluation of each mechanism under service resource-constrained conditions, Fig.~8 further illustrates the trends in average service quality and average task utility as the number of tasks increases. The experiments are conducted under two worker configurations: 15 workers (Figs.~8(a)-(b)) and  30 workers (Figs.~8(c)-(d)).

From Fig.~8, we observe that as the number of tasks increases, the average service quality and average task utility of methods including StagewiseTM$^3$atch, ConSpotTM$^3$atch, ConSpotTM$^3$atch\textbackslash TD, StagewiseTM$^3$atch\textbackslash Risk, and StagewiseTM$^3$atch\textbackslash TD generally exhibit a downward trend. This is due to intensified resource competition under a fixed number of workers, where the per-task available service decreases, thus degrading average performance.
Despite this, the proposed StagewiseTM$^3$atch mechanism consistently achieves the best performance in service quality, attributed to its stagewise task screening strategy and risk-aware trading design. These features enable the system to prioritize more feasible tasks under limited resources and improve completion quality through temporary worker recruitment.

Note that in Fig.~8(a), when the number of tasks ranges from 15 to 20, and in Fig.~8(c), when the number of tasks ranges from 20 to 25, the curves of StagewiseTM$^3$atch and ConSpotTM$^3$atch show slight declines and tend to stabilize. This phenomenon occurs because the number of tasks increases, some workers are able to handle multiple tasks simultaneously under permissible spatio-temporal constraints, thereby temporarily enhancing overall service efficiency and offsetting the negative impact brought by the number of tasks increases. However, as the task number continues to increase, workers' service capacity approaches saturation, leading to a further decline in service quality.
Moreover, StagewiseTM$^3$atch\textbackslash Risk does not incorporate risk analysis for tasks and workers, resulting in a significantly higher rate of worker defaults. Interestingly, due to the penalty compensation mechanism, task owners may receive higher utility despite poor service quality, thus explaining its relatively good performance in average task utility.

It is also important that, although all results are averaged over 100 Monte Carlo experiments to suppress random fluctuations, the performance curves of Quality\_M and Random\_M still exhibit noticeable volatility. This instability reflects the inherent limitations of these two mechanisms: Quality\_M relies solely on task utility for greedy worker selection, ignoring feasibility, constraints, and dynamic system evolution, leading to non-robust matches. Random\_M, in contrast, depends entirely on random assignments and is highly sensitive to initial conditions, resulting in inconsistent outcomes across trials. This instability becomes even more pronounced as the number of task continues to increase, further emphasizing the necessity of rational, capacity-aware, and risk-sensitive matching strategies.}

{\subsection{Influence of Payment Decrement $\triangle p$ on FT-M2M Performance}

\begin{table}[b!]
		{\centering
		\caption{Impact of payment decrement $\triangle p$ on FT-M2M matching}
		\vspace{0.5em}
		\begin{tabular}{|c|c|c|c|}
			\hline
			$\triangle p$ & \makecell[c]{Convergence\\ Rounds} & \makecell[c]{Expected \\tasks' utility} & \makecell[c]{Expected\\ workers' utility} \\
			\hline
			0.1    & 12       & 15.6      &  22.4 \\ \hline
			0.5    & 7       & 18.3       &  19.5 \\ \hline
			1.0    & 5       & 20.1       & 16.3  \\
			\hline
		\end{tabular}}
	\end{table}
Table~5 evaluates the impact of the payment decrement $\triangle p$ on the FT-M2M matching outcomes, from three key perspectives: the average expected utility of workers, the average expected utility of task owners, and the number of matching rounds required to reach stability. The experiment involves 20 tasks and 30 workers, with $\triangle p$ set to 0.1, 0.5, and 1.0, respectively.

As presented in Table~5, increasing $\triangle p$ results in a higher average expected utility for task owners and fewer matching rounds. This trend arises because a larger $\triangle p$ enables workers’ asked payments to reach their cost thresholds more rapidly, allowing task owners to recruit more workers at lower prices, thereby improving their overall utility.
In contrast, the average expected utility of workers declines as $\triangle p$ increases. This is due to the fact that smaller decrement rates give workers more opportunities to reach agreements at relatively higher asked payments, leading to better utility outcomes for the workers.}

\end{document}